\pdfoutput=1
\pdfoutput=1
\pdfoutput=1
\pdfoutput=1
\pdfoutput=1
\documentclass[11pt]{article}
\usepackage{booktabs}
\usepackage{subfloat}
\usepackage{amssymb}
\usepackage{amsmath}
\usepackage{amsthm}
\usepackage{array}
\usepackage{graphicx}
\usepackage{subfigure}
\usepackage{epstopdf}
\usepackage{cases}
\usepackage{color}
\usepackage{multirow}
\usepackage{picinpar}
\usepackage[colorlinks,
            linkcolor=red,
            anchorcolor=green,
            citecolor=blue
            ]{hyperref}
\usepackage{floatrow}%%多个表格放在同一行
\usepackage{enumitem}
\usepackage{fourier}
\begin{document}

\newcommand{\defi}{\stackrel{\Delta}{=}}
\newcommand{\A}{{\cal A}}
\newcommand{\B}{{\cal B}}
\newcommand{\U}{{\cal U}}
\newcommand{\G}{{\cal G}}
\newcommand{\cZ}{{\cal Z}}
\newcommand\one{\hbox{1\kern-2.4pt l }}
\newcommand{\Item}{\refstepcounter{Ictr}\item[\left(\theIctr\right)]}
\newcommand{\QQ}{\hphantom{MMMMMMM}}

\newtheorem{Theorem}{Theorem}[section]
\newtheorem{Lemma}{Lemma}[section]
\newtheorem{Corollary}{Corollary}[section]
\newtheorem{Remark}{Remark}[section]
\newtheorem{Example}{Example}[section]
\newtheorem{Proposition}{Proposition}[section]
\newtheorem{Property}{Property}[section]
\newtheorem{Assumption}{Assumption}[section]
\newtheorem{Definition}{Definition}[section]
\newtheorem{Construction}{Construction}[section]
\newtheorem{Condition}{Condition}[section]
\newtheorem{Exa}[Theorem]{Example}
\newcounter{claim_nb}[Theorem]
\setcounter{claim_nb}{0}
\newtheorem{claim}[claim_nb]{Claim}
\newenvironment{cproof}
{\begin{proof}
 [Proof.]
 \vspace{-3.2\parsep}}
{\renewcommand{\qed}{\hfill $\Diamond$} \end{proof}}
\newcommand{\erhao}{\fontsize{21pt}{\baselineskip}\selectfont}
\newcommand{\xiaoerhao}{\fontsize{18pt}{\baselineskip}\selectfont}
\newcommand{\sanhao}{\fontsize{15.75pt}{\baselineskip}\selectfont}
\newcommand{\sihao}{\fontsize{14pt}{\baselineskip}\selectfont}
\newcommand{\xiaosihao}{\fontsize{12pt}{\baselineskip}\selectfont}
\newcommand{\wuhao}{\fontsize{10.5pt}{\baselineskip}\selectfont}
\newcommand{\xiaowuhao}{\fontsize{9pt}{\baselineskip}\selectfont}
\newcommand{\liuhao}{\fontsize{7.875pt}{\baselineskip}\selectfont}
\newcommand{\qihao}{\fontsize{5.25pt}{\baselineskip}\selectfont}
\newcounter{Ictr}
\renewcommand{\theequation}{%\thesection.
\arabic{equation}}
\renewcommand{\thefootnote}{\fnsymbol{footnote}}

\def\A{\mathcal{A}}

\def\C{\mathcal{C}}

\def\V{\mathcal{V}}

\def\I{\mathcal{I}}

\def\Y{\mathcal{Y}}

\def\X{\mathcal{X}}

\def\J{\mathcal{J}}

\def\Q{\mathcal{Q}}

\def\W{\mathcal{W}}

\def\S{\mathcal{S}}

\def\T{\mathcal{T}}

\def\L{\mathcal{L}}

\def\M{\mathcal{M}}

\def\N{\mathcal{N}}
\def\R{\mathbb{R}}
\def\H{\mathbb{H}}

\title{}
\author{}
%%%%%%%%%%%%%%%%%%%%%%%%%%%%%%%%%%%%%%%%%%%%%%%%%%%%%%%%%%%%%%%%%%%%%%%%
\begin{center}
\topskip2cm
\LARGE{\bf Tuning parameter selection for the adaptive nuclear norm regularized trace regression}
\end{center}

\begin{center}
\renewcommand{\thefootnote}{\fnsymbol{footnote}}Pan Shang,  Lingchen Kong, Yiting Ma \footnote{Address: Pan Shang is  with Academy of Mathematics and Systems Science, Chinese Academy of Sciences, Beijing 100190, China; Lingchen Kong  and Yiting Ma are with the School of Mathematics and Statistics, Beijing Jiaotong University, Beijing, 100044, China.\\
E-mail: pshang@amss.ac.cn}\\
\today
%(June 14th, 2021)\\
\end{center}
\vskip4pt
\textbf{Abstract:} Regularized models have been applied in lots of areas, with high-dimensional data sets being popular. Because tuning parameter decides the theoretical performance and computational efficiency of the regularized models, tuning parameter selection is a basic and important issue.   We consider the tuning parameter selection for adaptive nuclear norm regularized trace regression, which achieves by the Bayesian information criterion (BIC). The proposed BIC is established with the help of an unbiased estimator of degrees of freedom. Under some regularized conditions, this BIC is proved to achieve the rank consistency of the tuning parameter selection. That is the model solution under selected tuning parameter converges to the true solution and has the same rank with that of the true solution in probability. Some numerical results are presented to evaluate the performance of the proposed BIC on tuning parameter selection.
\noindent \emph{MSCcodes(2020):} 62J12, 62H12, 90C25, 15A24
\section{Introduction}
To deal with high-dimensional data sets with special structures, there are plenties of works that focus on regularized models. These models can be unified as a minimization problem, whose objective function is composed with a loss function and some regularizer. So, existing works consider different loss functions and regularizer under different data settings. For instance, the famous Lasso\cite{tibshirani1996regression} and lots of similar models\cite{tibshirani2005sparsity,zou2005regularization,yuan2006model,
zou2006adaptive,zou2009adaptive,fan2001variable,zhang2010nearly} inducing sparsity, regularized matrix regression\cite{zhou2014regularized} and some other models inducing low-rank solution\cite{yuan2007dimension, rothman2010sparse,lu2012convex, 2014Multivariate, yuan2016degrees, zhu2020convex, wei2022high, zou2022estimation,mazumder2010spectral, zhou2014regularized, zhao2017trace, elsener2018robust, fan2019generalized, hamidi2022low} and so on. Undoubtedly, these regularized models have been applied to lots of areas, see e.g., bioinformatics\cite{zhao2010low}, engineering\cite{2010Statistical}, signal process\cite{2016An}, system identification\cite{2010Interior} and so on. Among these models, there is an interesting one called the adaptive nuclear norm regularized trace regression, which is proposed by Bach\cite{bach2008consistency} and achieves the rank consistency. Therefore, we focus on this model in the paper.

For regularized models, their loss function and regularizer are summed together by the so-called tuning parameter. It seems that tuning parameter is insignificant, while the selection of this parameter decides the theory performance and computational effect of these models. Therefore, there are some tools can be used to select this parameter\cite{wu2020survey}, such as Akaike information criterion (AIC,\cite{akaike1970statistical}), Bayesian information criterion (BIC,\cite{schwarz1978estimating,wang2007unified,
wang2009shrinkage,wang2007tuning, wang2011consistent,kim2012consistent,fan2013tuning,hirose2013tuning,sun2013consistent,yaguang2018consistent,abbruzzo2019selecting}), cross-validation (\cite{homrighausen2013lasso, homrighausen2017risk,lei2020cross,datta2020note, chetverikov2021cross}), and so on. Here, we review some related works about existing tuning parameter selection methods for regularized models, which include regularized liner models with the unknown variable being in vector form and regularized matrix models with matrix as variable, respectively.  The following results mainly focus on tuning parameter selection for regularized linear models. Wang et al. \cite{wang2007tuning} established a BIC for SCAD under the case that the sample size $n$ and feature size $p$ satisfying $n>p$ and $p$ fixed. Wang et al. \cite{wang2009shrinkage} considered the BIC for adaptive Lasso under $n>p$ and $p$ diverged. Wang and Zhu \cite{wang2011consistent} builded up a high dimensional BIC (HBIC) for adaptive elastic net under $n<p$ and $p$ diverged. Kim et al. \cite{kim2012consistent} builded a generalized information criterion for SCAD and MCP under different data settings, including $n<p$ with $p$ fixed and diverged. Hirose et al. \cite{hirose2013tuning} considered the extended  regularized information criterion for adaptive Lasso. Li et al. \cite{yaguang2018consistent} established the regularized information criterion for adaptive group regularized generalized linear models. These proposed information criteria all were proved to have the consistency property, that is the selected tuning parameter will lead to a solution with the same active index as the true one in probability, when the sample size increases. So, these regularized linear models all have an important property, which is the selection consistency. In addition, there are a few tuning parameter selection results for regularized matrix regression. Zhou and Li \cite{zhou2014regularized} proposed the AIC and BIC for regularized matrix regression under the orthogonal assumption of prediction matrix. Yuan \cite{yuan2016degrees} calculated the degree of freedom of regularized matrix regression, which is the basic for some tuning parameter selection methods. However, these results of regularized matrix regression have no theoretical guarantees of consistency.

In this paper, we tends to select the tuning parameter for adaptive nuclear norm regularized trace regression proposed by \cite{bach2008consistency}. The core reason we consider this model is its rank consistency, i.e., the model solution achieved the true rank in probability with the sample size increasing. In addition, this model also includes the regularized matrix regression models in \cite{zhou2014regularized} and \cite{yuan2006model} as special cases. We first establish some basic results, including the subdifferential and conjugate function of the adaptive nuclear norm. These results help to analyze the model and provide tools to establish its degrees of freedom, which is en essential element to set up the tuning parameter selection criterion. Based on the definition in \cite{efron2004estimation}, we calculate the unbiased estimation of degrees of freedom, according to the model optimality and basic results of the model. Although the degrees of freedom seems like a complex form, we prove that this result includes the existing result in \cite{zhou2014regularized} as a special case. With the help of the estimator of the degrees of freedom, we propose a BIC to select the tuning parameter for the adaptive nuclear norm regularized trace regression. The proposed BIC is proved to achieve the rank selection consistency in probability, under some regularized conditions. That is, the tuning parameter selected by BIC achieves selection consistency and rank selection consistency simultaneously. On numerical experiments, we compare the proposed BIC with AIC, AICc and cross validation for tuning parameter selection on some simulation data and real data sets, which illustrates the efficiency of our criterion.

This paper is organized as follows. We review the adaptive nuclear norm regularized trace regression and present some basic results in Section 2. In Section 3, we establish the degrees of freedom and BIC with theoretical guarantees. Some numerical experiments and conclusion are showed in Section 4 and Section 5, respectively.

\emph{Notations}:
Let $M\in\mathbb{R}^{p_{1}\times p_{2}}$ be a matrix. vec$(M)$ denotes the vector in $\mathbb{R}^{p_{1}p_{2}}$ obtained by stacking its columns into a single vector. For any $j\in\{1,2,\cdots,p_{1}\}$ and $k\in\{1,2,\cdots,p_{2}\}$, $M_{:,j}$ means the $j_{th}$ column  and $M_{k,:}$ means the $k_{th}$ row of $M$. Suppose $M$ has a singular value decomposition with nonincreasing singular values $\sigma_{1}(M)\geq \cdots \geq \sigma_{r}(M)>0$ and $r\leq\rm{\min\left\{\textit{p}_{1}, \textit{p}_{2}\right\}}$ is the rank of $M$. There are some related  norms with singular values of $M$. The Frobenius norm $\|\cdot\|_{F}$ is defined as $\|M\|_{F}=\sqrt{\sum_{i=1}^{p_{1}}\sum _{j=1}^{p_{2}}M_{i,j}^{2}}=\sqrt{\sigma^{2}_{1}(M)+\cdots+\sigma^{2}_{r}(M)}$. The nuclear norm $\|\cdot\|_{*}$ is  the sum of non-zero singular values, i.e., $\|M\|_{*}=\sum_{i=1}^{r}{\sigma_{i}(M)}$. The spectral norm $\|\cdot\|_{2}$ is the largest singular value, i.e., $\|M\|_{2}=\sigma_{1}(M)$.  A symmetric matrix $M\in\mathbb{R}^{p\times p}$ is called positive semidefinite (positive definite), denoted as $M\succeq0$ $(M\succ0)$, if $\textbf{\textit{x}}^{\top}M\textbf{\textit{x}}\geq0$ $ (\textbf{\textit{x}}^{\top}M\textbf{\textit{x}}>0)$ holds for any $0\neq\textbf{\textit{x}}\in\mathbb{R}^{p}$. Any $M\in\mathbb{O}^{p}$ means that $M\in\mathbb{R}^{p\times p}$ and $M^{\top}M=I_{p}$. For any vector $\textbf{\textit{x}}\in\mathbb{R}^{p}$,
the norm $\|\cdot\|$ is defined as $\|\textbf{\textit{x}}\|=\sqrt{\sum_{i=1}^{p}x_{i}^{2}}$. The notation $\textbf{\textit{x}}\in \mathbb{R}^{p}_{++}$ means all elements of $\textbf{\textit{x}}$ are positive. The indictor function of a set $\mathcal{A}$ is denoted as $\delta_{\mathcal{A}}(\cdot)$, which means the value of this function is zero when variable in the set $\mathcal{A}$ and is $\infty$ otherwise. For any index set $I$, its cardinal number  is denoted as $|I|$ that counts the number in the index set $I$, and its complementary set is denoted as $I^{c}$.  In this paper, the notation 0 may represent scalar, vector and matrix, which can be inferred from the context.
\section{Preliminaries}
This section presents the adaptive nuclear norm regularized trace regression with its analysis and show some optimal results of this model, which are foundations of the tuning parameter selection.
\subsection{Model Analysis}
\label{2.1}
The statistical model of the trace regression is
\begin{eqnarray*}
y = \left\langle X,B^{*}\right\rangle+\epsilon,
\end{eqnarray*}
where $X\in \mathbb{R}^{p_{1}\times p_{2}}$ is the predictor, $y\in \mathbb{R}$ is the responsor, $\epsilon\in \mathbb{R}$ is a random error and $B^{*}\in \mathbb{R}^{p_{1}\times p_{2}}$ is the true coefficient. By sampling $n$ times, we get
\begin{eqnarray*}
y_{i} = \left\langle X_{i},B^{*}\right\rangle+\epsilon_{i}, \quad i=1,2,\cdots,n.
\end{eqnarray*}
Let $\mathcal{X}=\left(\rm{vec}\left(\textit{X}_{1}\right),\cdots,\rm{vec}\left(\textit{X}_{\textit{n}}\right)\right)^{\top}\in\mathbb{R}^{\textit{n}\times \textit{p}_{1} \textit{p}_{2}}$, $\textbf{y}=(y_{1},\cdots,y_{n})^{\top}\in\mathbb{R}^{n}$ and $\varepsilon=(\epsilon_{1},\cdots,\epsilon_{n})^{\top}\in\mathbb{R}^{n}$. These sample models can be written as
\begin{eqnarray*}
\textbf{y} = \mathcal{X}\rm{vec}\left(\textit{B}^{*}\right)+\varepsilon.
\end{eqnarray*}
To estimate the unknown matrix $B^{*}$, there are some literatures focus on the nuclear norm regularized trace regression, such as Koltchinskii et al. \cite{2011Nuclear},  Negahban and Wainwright \cite{negahban2011estimation}, Zhou and Li \cite{zhou2014regularized} and so on, which is
\begin{eqnarray*}
\underset{B\in \mathbb{R}^{p_{1}\times p_{2}}}\min\left\{\frac{1}{2n}\sum\limits_{i=1}^{n}\left(y_{i}-\left\langle X_{i},B\right\rangle\right)^{2}+\lambda\|B\|_{*}\right\}.
\end{eqnarray*}
Here, $\lambda>0$ is the tuning parameter. However, this model  can hardly achieve the rank consistent solution in high-dimensional case as in Bach \cite{bach2008consistency}. So, an adaptive version  was proposed as
\begin{eqnarray}\label{eq1}
\underset{B\in \mathbb{R}^{p_{1}\times p_{2}}}\min\left\{\frac{1}{2n}\sum\limits_{i=1}^{n}\left(y_{i}-\left\langle X_{i},B\right\rangle\right)^{2}+\lambda\|W_{1}BW_{2}\|_{*}\right\},
\end{eqnarray}
where $W_{1}$ and $W_{2}$ are given weight matrixes based on the solution of least square trace regression. The detailed results of these matrixes are showed as follows.
Let
$$\hat{B}_{LS}=\underset{B\in \mathbb{R}^{p_{1}\times p_{2}}}{\arg\min}\left\{\frac{1}{2n}\sum\limits_{i=1}^{n}\left(y_{i}-\left\langle X_{i},B\right\rangle\right)^{2}\right\}.$$
Suppose $\hat{B}_{LS}$ has full rank. Let the singular value decomposition of $\hat{B}_{LS}$ be $\hat{B}_{LS}=U_{LS}$Diag$\left(s^{LS}\right)V^{\top}_{LS}$,
where $U_{LS}\in\mathbb{O}^{p_{1}}$, $V_{LS}\in\mathbb{O}^{p_{2}}$, $s^{LS}\in \mathbb{R}^{\min\{p_{1},p_{2}\}}$
and Diag$\left(s^{LS}\right)$ is a matrix whose diagonal vector is the singular vector $s^{LS}$. $s^{LS}$ can be completed by $n^{-1/2}$ to reach dimensions $p_{1}$ or $p_{2}$. Then,
\begin{center}
$W_{1}=U_{LS}\rm{Diag}\left(\textit{s}^{\textit{LS}}\right)^{-\gamma}$$\textit{U}^{\top}_{\textit{LS}}$, $W_{2}=V_{LS}\rm{Diag}\left(\textit{s}^{\textit{LS}}\right)^{-\gamma}$$\textit{V}^{\top}_{\textit{LS}}$, $\gamma\in(0,1]$.
\end{center}
Based on expressions  of  $W_{1}$ and $W_{2}$, they are obvious symmetric matrixes, i.e., $W^{\top}_{1}=W_{1}$ and
$W^{\top}_{2}=W_{2}$. In addition, $W^{-1}_{1}=U_{LS}\rm{Diag}\left(\textit{s}^{\textit{LS}}\right)^{\gamma}$$\textit{U}^{\top}_{\textit{LS}}$ and $W^{-1}_{2}=V_{LS}\rm{Diag}\left(\textit{s}^{\textit{LS}}\right)^{\gamma}$$\textit{V}^{\top}_{\textit{LS}}$.

%For any $\lambda>0$, denote the solution of (\ref{eq1}) as $\hat{B}(\lambda)$. It is well known that the tuning parameter selection is an essential and important issue for (\ref{eq1}).
%The reasons can be summarized as follows. Theoretically, the estimation bound of (\ref{eq1}) relies on the choice of $\lambda$, which is a criterion to evaluate the performance of models. Computationally, the computational time varies with different $\lambda$.
\subsection{Basic Results}
\label{2.2}
In this section, some related results are  proposed and reviewed, which provide theoretical guarantees for the rest of this paper. The following definitions are from Rockafellar \cite{rockafellar2015convex}.
\begin{Definition}
Let $f:\mathbb{R}^{p_{1}\times p_{2}}\rightarrow (-\infty,+\infty]$ be a proper closed convex function and let $M\in \mathbb{R}^{p_{1}\times p_{2}}$.
%A matrix $G\in \mathbb{R}^{p_{1}\times p_{2}}$ is called a subgradient of $f$ at $M$ if
%\begin{center}
%$f(N)\geq f(M)+\langle G,N-M\rangle$,  $\forall N\in \mathbb{R}^{p_{1}\times p_{2}}$.
%\end{center}
The  subdifferential of $f$ at $M$ is denoted by $\partial f(M)$ and
\begin{center}
$\partial f(M)=\left\{G\in \mathbb{R}^{p_{1}\times p_{2}}: f(N)\geq f(M)+\langle G,N-M\rangle, \forall N\in \mathbb{R}^{p_{1}\times p_{2}} \right\}$.
\end{center}
\end{Definition}
According to this definition, we compute the subdifferential of $f(M)=\|W_{1}MW_{2}\|_{*}$ as follows.
\begin{Proposition}\label{pro2.1}
For any $M\in\mathbb{R}^{p_{1}\times p_{2}}$, let $r$ be the rank of $W_{1}MW_{2}$. The subdifferential of $\|W_{1}MW_{2}\|_{*}$ is
\begin{center}
$\partial \|W_{1}MW_{2}\|_{*}=\left\{ W_{1}\left(U_{r}V^{\top}_{r}+N\right)W_{2}:W_{1}MW_{2}=U_{r}\rm{Diag}\left(\sigma_{\textit{r}}\right)\textit{V}^{\top}_{\textit{r}},\|\textit{N}\|_{2}\leq1,
\textit{U}^{\top}_{\textit{r}}\textit{N}=0,\textit{NV}_{\textit{r}}=0\right\},$
\end{center}
where $U_{r}\in\mathbb{R}^{p_{1}\times r}$, $V_{r}\in\mathbb{R}^{p_{2}\times r}$, $\sigma_{r}\in\mathbb{R}^{r}_{++}$ and $N\in\mathbb{R}^{p_{1}\times p_{2}}$.
\end{Proposition}
\begin{proof}
It is clear that $f(M)=h(g(M))$ with $g(M)=W_{1}MW_{2}$ and $h(M)=\lambda\|M\|_{*}$. Because $g$ is a linear transformation,
according to \cite[Theorem 3.43]{BECK2017FIRST},
\begin{center}
$\partial f(M)=g^{\top}\left(\partial_{g(M)}h(g(M))\right)$
\end{center}
with $g^{\top}$ being the adjoint operator of $g$.

According to the definition of the adjoint operator, $\left\langle g(M),N\right\rangle=\left\langle M,g^{\top}(N)\right\rangle$, which leads to
\begin{center}
 $g^{\top}(N)=W^{\top}_{1}NW^{\top}_{2}=W_{1}NW_{2}$.
\end{center}

Based on  the singular value decomposition of $W_{1}MW_{2}$ and the subdifferential of the nuclear norm in Watson \cite{watson1992characterization}, it holds that
\begin{center}
$\partial_{g(M)}h(g(M))=\left\{U_{r}V^{\top}_{r}+N:\|N\|_{2}\leq1,U^{\top}_{r}N=0,NV_{r}=0\right\}$.
\end{center}
Combing these results, the desired result can be obtained.
\end{proof}

There are some basic calculation rule in \cite{horn2012matrix,magnus2019matrix} that are foundations of our result, which we review them in the next proposition.
\begin{Proposition}\label{pro2.2}
Let $A,B,C,D$ be any  matrix in compatible dimensions, the following claims hold.
\begin{enumerate}
\item[i)]{vec$(ABC)=\left(C^{\top}\bigotimes A\right)$vec$(B)$, where $\bigotimes$ means the Kronecker product.}
\item[ii)]{vec$\left(A^{\top}\right)=K$vec$(A)$, where $K$ is the permutation matrix.}
\item[iii)]{$(A\bigotimes B)(C\bigotimes D)=AC\bigotimes BD$.}
\item[iv)]{$\frac{\partial vec(ABC)}{\partial vec^{\top}(B)}=C^{\top}\bigotimes A$.}
\end{enumerate}
\end{Proposition}

\section{Bayesian Information Criterion}
\label{3}

In this section, we tend to build up a Bayesian Information Criterion (BIC) and prove its consistency to select the appropriate tuning parameter for (\ref{eq1}).

According to the definition of traditional BIC, a key basis is degrees of freedom of the model. So, we calculate degrees of freedom of (\ref{eq1}) in the next theorem.
\begin{Theorem}\label{lemma df}
For any $\lambda \geq0$, let $r$ be the rank of $W_{1}\hat{B}(\lambda)W_{2}$. Suppose the singular value decomposition of $W_{1}\hat{B}(\lambda)W_{2}$ is $W_{1}\hat{B}(\lambda)W_{2}=U_{r}\rm{Diag}\left(\textit{b}_{\textit{r}}\right)\textit{V}_{\textit{r}}^{\top}$ with $U_{r}\in\mathbb{R}^{p_{1}\times r}$, $V_{r}\in\mathbb{R}^{p_{2}\times r}$ and $b_{r}\in\mathbb{R}^{r}_{++}$. An unbiased estimation of degrees of freedom of (1) under $\lambda$  satisfies that
\begin{center}
$\hat{\rm{df}}_{\lambda}=\frac{1}{\textit{n}}\sum\limits_{\textit{k}=1}^{\textit{n}}
\rm{vec}\left(\textit{X}_{\textit{k}}\right)^{\top}$$M_{r}^{+}\rm{vec}\left(\textit{X}_{\textit{k}}\right).$
\end{center}
Here, $M_{r}=\frac{1}{\textit{n}}\sum\limits_{\textit{k}=1}^{\textit{n}}\rm{vec}\left(\textit{X}_{\textit{k}}\right)
\rm{vec}\left(\textit{X}_{\textit{k}}\right)^{\top}$$+\lambda M^{(1)}_{r}+M^{(2)}_{r}$ is assumed full rank or symmetric, with
$M^{(1)}_{r}$ and $M^{(2)}_{r}$ defined in Lemma \ref{lemma3.1} and Lemma \ref{lemma3.2}, respectively.
%\begin{center}
%P$\left(\hat{\rm{df}}_{\lambda}=\underset{n\rightarrow\infty}{\rm{lim}}\frac{1}{\textit{n}}\sum\limits_{\textit{k}=1}^{\textit{n}}
%\rm{vec}(\textit{X}_{\textit{k}})^{\top}\left(\Sigma+\lambda\textit{A}_{r}\right)^{-1}\rm{vec}(\textit{X}_{\textit{k}})
%\right)\rightarrow1$, when  $n\rightarrow\infty$.
%\end{center}
%Here,
%$A_{r}= \left(\textit{W}_{2}\textit{V}_{\textit{r}}\rm{Diag}(\textit{b}_{\textit{r}})^{-1}\textit{V}^{\top}_{\textit{r}}\textit{W}_{2}\right)\bigotimes \textit{W}^{2}_{1}+\lambda \textit{W}^{2}_{2}\bigotimes
%\left(\textit{W}_{1}\textit{U}_{\textit{r}}\rm{Diag}(\textit{b}_{\textit{r}})^{-1}\textit{U}^{\top}_{\textit{r}}\textit{W}_{1}\right).$
\end{Theorem}
\begin{proof}
As the result in Stein \cite{stein1981estimation}, an unbiased estimation of degrees of freedom of (1) under $\lambda$ is
\begin{center}
$\hat{\rm{df}}_{\lambda}=\sum\limits_{k=1}^{n}\frac{\partial \hat{y}_{k}}{\partial y_{k}}=\sum\limits_{k=1}^{n}\left\langle X_{k},\frac{\partial \hat{B}(\lambda)}{\partial y_{k}}\right\rangle=\sum\limits_{k=1}^{n}\left\langle \rm{vec}\left(\textit{X}_{\textit{k}}\right),\frac{\partial \rm{vec}\left(\hat{\textit{B}}(\lambda)\right)}{\partial y_{k}}\right\rangle$.
\end{center}
According to the optimal condition of (1) and the subdifferential of $\|W_{1}BW_{2}\|_{*}$ in Proposition \ref{pro2.1}, there is a matrix $N$ such that $\|N\|_{2}\leq1,U^{\top}_{r}N=0,NV_{r}=0$ and
\begin{eqnarray}\label{eq}
-\frac{1}{n}\sum\limits_{i=1}^{n}\left(y_{i}-\left\langle X_{i},\hat{B}(\lambda)\right\rangle\right)X_{i}+\lambda W_{1}G\left(\hat{B}(\lambda)\right)W_{2}+\lambda W_{1}NW_{2}=0,
\end{eqnarray}
where $G\left(\hat{B}(\lambda)\right)=U_{r}V^{\top}_{r}$. Vectorizing this equality and using $i)$ in Proposition \ref{pro2.2}, we get that
\begin{eqnarray} \label{eqop}
\frac{1}{n}\sum\limits_{\textit{i}=1}^{\textit{n}}y_{i}\rm{vec}\left(\textit{X}_{\textit{i}}\right)
=\frac{1}{\textit{n}}\sum\limits_{\textit{i}=1}^{\textit{n}}\rm{vec}\left(\textit{X}_{\textit{i}}\right)
\rm{vec}\left(\textit{X}_{\textit{i}}\right)^{\top}\rm{vec}\left(\hat{\textit{B}}(\lambda)\right)+\lambda \left(\textit{W}_{2} \bigotimes \textit{W}_{1}\right)\rm{vec}\left(\textit{G}\left(\hat{\textit{B}}(\lambda)\right)\right)+\lambda\left(\textit{W}_{2} \bigotimes \textit{W}_{1}\right)\rm{vec}(\textit{N}).
\end{eqnarray}
The result of the partial differential of the equation (\ref{eqop}) with respect of $y_{k}$ is
\begin{equation}\label{eqnew}
\begin{aligned}
\frac{1}{n}\rm{vec}(\textit{X}_{\textit{k}})&=\left[\frac{1}{\textit{n}}
\sum\limits_{\textit{i}=1}^{\textit{n}}\rm{vec}\left(\textit{X}_{\textit{i}}\right)
\rm{vec}\left(\textit{X}_{\textit{i}}\right)^{\top}
+\lambda \left(\textit{W}_{2}\bigotimes \textit{W}_{1}\right)\frac{\partial \rm{vec}\left(\textit{G}(\hat{\textit{B}}(\lambda))\right)}{\partial \rm{vec}^{\top}\left(\hat{\textit{B}}(\lambda)\right)}+\lambda \left(\textit{W}_{2}\bigotimes \textit{W}_{1}\right)\frac{\partial \rm{vec}(N)}{\partial \rm{vec}^{\top}\left(\hat{\textit{B}}(\lambda)\right)}\right]\frac{\partial \rm{vec}\left(\hat{\textit{B}}(\lambda)\right)}{\partial y_{k}}\\
&:=\left[\frac{1}{\textit{n}}\sum\limits_{\textit{i}=1}^{\textit{n}}
\rm{vec}\left(\textit{X}_{\textit{i}}\right)\rm{vec}\left(\textit{X}_{\textit{i}}\right)^{\top}+\lambda \textit{M}^{(1)}_{\textit{r}}+\textit{M}^{(2)}_{\textit{r}}\right] \frac{\partial \rm{vec}\left(\hat{\textit{B}}(\lambda)\right)}{\partial y_{k}},
\end{aligned}
\end{equation}
where
\begin{center}
$M^{(1)}_{r}=\left(\textit{W}_{2}\bigotimes \textit{W}_{1}\right)\frac{\partial \rm{vec}\left(\textit{G}\left(\hat{\textit{B}}(\lambda)\right)\right)}{\partial \rm{vec}^{\top}\left(\hat{\textit{B}}(\lambda)\right)}$,
  $M^{(2)}_{r}=\lambda\left(\textit{W}_{2}\bigotimes \textit{W}_{1}\right)\frac{\partial \rm{vec}(\textit{N})}{\partial \rm{vec}^{\top}\left(\hat{\textit{B}}(\lambda)\right)}$.
\end{center}

Denote $M_{r}=\frac{1}{\textit{n}}\sum\limits_{\textit{i}=1}^{\textit{n}}
\rm{vec}\left(\textit{X}_{\textit{i}}\right)\rm{vec}\left(\textit{X}_{\textit{i}}\right)^{\top}
+\lambda \textit{M}^{(1)}_{\textit{r}}+\textit{M}^{(2)}_{\textit{r}}$. Then, the equation (\ref{eqnew}) is transformed as
\begin{eqnarray}{\label{eqnewtr}}
\frac{1}{n}\rm{vec}\left(\textit{X}_{\textit{k}}\right)=\textit{M}_{\textit{r}}\frac{\partial \rm{vec}\left(\hat{\textit{B}}(\lambda)\right)}{\partial y_{k}}.
\end{eqnarray}

To get the result of $\hat{\rm{df}}_{\lambda}$, we assume that $M_{r}$ has full rank or it is a symmetric matrix. The results are illustrated as follows.

If $M_{r}$ has full rank, it is obvious that $\frac{\partial \rm{vec}\left(\hat{\textit{B}}(\lambda)\right)}{\partial y_{k}}=\frac{1}{n} M^{-1}_{r}\rm{vec}\left(\textit{X}_{\textit{k}}\right)$ and
$\hat{\rm{df}}_{\lambda}=\frac{1}{n}\sum\limits_{\textit{k}=1}^{\textit{n}}
\rm{vec}\left(\textit{X}_{\textit{k}}\right)^{\top}\textit{M}^{-1}_{\textit{r}}
\rm{vec}\left(\textit{X}_{\textit{k}}\right)$.

If $M_{r}$ is symmetric, suppose the eigenvalue decomposition of $M_{r}$ is
\begin{center}
$M_{r}=V\begin{pmatrix}\rm{Diag}\left(\sigma\left(\textit{M}_{\textit{r}}\right)\right),&0\\0,&0\end{pmatrix}V^{\top}$,
\end{center}
where $V=\left(\textbf{v}_{1},\cdots,\textbf{v}_{r_{M}},\textbf{v}_{r_{M}+1},
\cdots,\textbf{v}_{p_{1}p_{2}}\right)=\left(V_{r_{M}},V_{r_{M}^{\bot}}\right)\in\mathbb{O}^{p_{1}p_{2}}$, $V^{\top}_{r_{M}^{\bot}}V_{r_{M}}=0$ and $\sigma\left(M_{r}\right)\in\mathbb{R}^{r_{M}}$ with $r_{M}$ being the rank of $M_{r}$. Then, the pseudo inverse matrix of $M_{r}$ is
\begin{center}
$M^{+}_{r}=V\begin{pmatrix}\rm{Diag}\left(\sigma\left(\textit{M}_{\textit{r}}\right)\right)^{-1},&0\\0,&0\end{pmatrix}V^{\top}$.
\end{center}
Based on the equation (\ref{eqnewtr}),  we have
\begin{equation}\label{eqimp}
\begin{aligned}
\frac{1}{n}\textit{V}_{\textit{r}_{\textit{M}}^{\bot}}^{\top}\rm{vec}\left(\textit{X}_{\textit{k}}\right)
&=\textit{V}_{\textit{r}_{\textit{M}}^{\bot}}^{\top}\textit{M}_{\textit{r}}\frac{\partial \rm{vec}\left(\hat{\textit{B}}(\lambda)\right)}{\partial y_{k}}
=\textit{V}_{\textit{r}_{\textit{M}}^{\bot}}^{\top}V
\begin{pmatrix}\rm{Diag}\left(\sigma\left(\textit{M}_{\textit{r}}\right)\right),&0\\0,&0\end{pmatrix}
V^{\top}\frac{\partial \rm{vec}\left(\hat{\textit{B}}(\lambda)\right)}{\partial y_{k}}\\
%&=\textit{V}_{\textit{r}_{\textit{M}}^{\bot}}^{\top}V\begin{pmatrix}\rm{Diag}(\sigma(\textit{M}_{\textit{r}})),&0\\0,&0\end{pmatrix}
%\begin{pmatrix}V_{r_{M}}^{\top}\frac{\partial \rm{vec}(\hat{\textit{B}}(\lambda))}{\partial y_{k}}\\
%V_{r_{M}^{\bot}}^{\top}\frac{\partial \rm{vec}(\hat{\textit{B}}(\lambda))}{\partial y_{k}}\end{pmatrix}\\
&=\textit{V}_{\textit{r}_{\textit{M}}^{\bot}}^{\top}V\begin{pmatrix}\rm{Diag}\left(\sigma\left(\textit{M}_{\textit{r}}\right)\right)\textit{V}_{\textit{r}_{\textit{M}}}^{\top}
\frac{\partial \rm{vec}\left(\hat{\textit{B}}(\lambda)\right)}{\partial y_{k}}\\
0\end{pmatrix}=\begin{pmatrix}\textit{V}_{\textit{r}_{\textit{M}}^{\bot}}^{\top}\textit{V}_{\textit{r}_{\textit{M}}},
\textit{V}_{\textit{r}_{\textit{M}}^{\bot}}^{\top}\textit{V}_{\textit{r}_{\textit{M}}^{\bot}}
\end{pmatrix}
\begin{pmatrix}\rm{Diag}\left(\sigma\left(\textit{M}_{\textit{r}}\right)\right)\textit{V}_{\textit{r}_{\textit{M}}}^{\top}\frac{\partial \rm{vec}\left(\hat{\textit{B}}(\lambda)\right)}{\partial y_{k}}\\
0\end{pmatrix}\\
&=\begin{pmatrix}0,I_{p_{1}p_{2}-r_{M}}\end{pmatrix}
\begin{pmatrix}\rm{Diag}\left(\sigma\left(\textit{M}_{\textit{r}}\right)\right)\textit{V}_{\textit{r}_{\textit{M}}}^{\top}\frac{\partial \rm{vec}\left(\hat{\textit{B}}(\lambda)\right)}{\partial y_{k}}\\
0\end{pmatrix}=0.
\end{aligned}
\end{equation}
Combining the fact that $M^{+}_{r}M_{r}=V_{r_{M}}V_{r_{M}}^{\top}$ and (\ref{eqimp}), we have
\begin{equation*}
\begin{aligned}
\hat{\rm{df}}_{\lambda}&=\sum\limits_{\textit{k}=1}^{\textit{n}}
\left\langle\rm{vec}\left(\textit{X}_{\textit{k}}\right),\frac{\partial \rm{vec}\left(\hat{\textit{B}}(\lambda)\right)}{\partial y_{k}}\right\rangle\\
&=\sum\limits_{\textit{k}=1}^{\textit{n}}\rm{vec}\left(\textit{X}_{\textit{k}}\right)^{\top}
\textit{V}\textit{V}^{\top}\frac{\partial \rm{vec}\left(\hat{\textit{B}}(\lambda)\right)}{\partial y_{k}}=\sum\limits_{\textit{k}=1}^{\textit{n}}\rm{vec}\left(\textit{X}_{\textit{k}}\right)^{\top}\textit{M}^{+}_{\textit{r}}\textit{M}_{\textit{r}}\frac{\partial \rm{vec}\left(\hat{\textit{B}}(\lambda)\right)}{\partial y_{k}}
+\sum\limits_{\textit{k}=1}^{\textit{n}}\rm{vec}\left(\textit{X}_{\textit{k}}\right)^{\top}
\textit{V}_{\textit{r}_{\textit{M}}^{\bot}}\textit{V}_{\textit{r}_{\textit{M}}^{\bot}}^{\top}
\frac{\partial \rm{vec}\left(\hat{\textit{B}}(\lambda)\right)}{\partial y_{k}}\\
&\underset{(\ref{eqnewtr})}{=}\frac{1}{n}\sum\limits_{\textit{k}=1}^{\textit{n}}\rm{vec}
\left(\textit{X}_{\textit{k}}\right)^{\top}\textit{M}^{+}_{\textit{r}}
\rm{vec}\left(\textit{X}_{\textit{k}}\right)
+\sum\limits_{\textit{k}=1}^{\textit{n}}\rm{vec}\left(\textit{X}_{\textit{k}}\right)^{\top}
\textit{V}_{\textit{r}_{\textit{M}}^{\bot}}\textit{V}_{\textit{r}_{\textit{M}}^{\bot}}^{\top}\frac{\partial \rm{vec}\left(\hat{\textit{B}}(\lambda)\right)}{\partial y_{k}}\\
&=\frac{1}{n}\sum\limits_{\textit{k}=1}^{\textit{n}}\rm{vec}\left(\textit{X}_{\textit{k}}\right)^{\top}\textit{M}^{+}_{\textit{r}}
\rm{vec}\left(\textit{X}_{\textit{k}}\right)
+
\sum\limits_{\textit{k}=1}^{\textit{n}}\left\langle\textit{V}_{\textit{r}_{\textit{M}}^{\bot}}^{\top}\rm{vec}\left(\textit{X}_{\textit{k}}\right),
\textit{V}_{\textit{r}_{\textit{M}}^{\bot}}^{\top}\frac{\partial \rm{vec}\left(\hat{\textit{B}}(\lambda)\right)}{\partial y_{k}}\right\rangle\\
&\underset{(\ref{eqimp})}{=}\frac{1}{n}\sum\limits_{\textit{k}=1}^{\textit{n}}
\rm{vec}\left(\textit{X}_{\textit{k}}\right)^{\top}\textit{M}^{+}_{\textit{r}}
\rm{vec}\left(\textit{X}_{\textit{k}}\right).
\end{aligned}
\end{equation*}
These two results leads to Theorem \ref{lemma df}.
\end{proof}
Next, we give detailed results of $M^{(1)}_{r}$ and  $M^{(2)}_{r}$ in following lemmas.
\begin{Lemma}\label{lemma3.1}
For any $\lambda\geq0$, let $r$ be the rank of $W_{1}\hat{B}(\lambda)W_{2}$. Suppose the singular value decomposition of $W_{1}\hat{B}(\lambda)W_{2}$ is $W_{1}\hat{B}(\lambda)W_{2}=U_{r}\rm{Diag}\left(\textit{b}_{\textit{r}}\right)\textit{V}_{\textit{r}}^{\top}$ with $U_{r}\in\mathbb{R}^{p_{1}\times r}$, $V_{r}\in\mathbb{R}^{p_{2}\times r}$ and $b_{r}\in\mathbb{R}^{r}_{++}$. Then,
\begin{equation}\label{eqn}
\begin{aligned}
M^{(1)}_{r}
&=\left(W_{2}\bigotimes W_{1}\right)\frac{\partial \rm{vec}\left(\textit{G}(\hat{\textit{B}}(\lambda))\right)}{\partial \rm{vec}^{\top}\left(\hat{\textit{B}}(\lambda)\right)}=
\left[W_{2}V_{r}\rm{Diag}\left(\textit{b}_{\textit{r}}\right)^{-1}\textit{V}^{\top}_{\textit{r}}\textit{W}_{2}\right]\bigotimes W^{2}_{1}+W^{2}_{2}\bigotimes\left[W_{1}U_{r}\rm{Diag}\left(\textit{b}_{\textit{r}}\right)^{-1}\textit{U}^{\top}_{\textit{r}}\textit{W}_{1}\right].
\end{aligned}
\end{equation}
\end{Lemma}
\begin{proof}
Note that $M^{(1)}_{r}=\left(\textit{W}_{2}\bigotimes \textit{W}_{1}\right)\frac{\partial \rm{vec}\left(\textit{G}\left(\hat{\textit{B}}(\lambda)\right)\right)}{\partial \rm{vec}^{\top}\left(\hat{\textit{B}}(\lambda)\right)}$. It is obvious that
\begin{eqnarray}\label{eqs}
\frac{\partial \rm{vec}\left(\textit{G}\left(\hat{\textit{B}}(\lambda)\right)\right)}{\partial \rm{vec}^{\top}\left(\hat{\textit{B}}(\lambda)\right)}
= \frac{\partial \rm{vec}\left(\textit{G}\left(\hat{\textit{B}}(\lambda)\right)\right)}{\partial \rm{vec}^{\top}\left(\textit{U}_{\textit{r}}\right)}\cdot\frac{\partial \rm{vec}\left(\textit{U}_{\textit{r}}\right)}{\partial \rm{vec}^{\top}\left(\hat{\textit{B}}(\lambda)\right)}
+\frac{\partial \rm{vec}\left(\textit{G}\left(\hat{\textit{B}}(\lambda)\right)\right)}{\partial \rm{vec}^{\top}\left(\textit{V}^{\top}_{\textit{r}}\right)}\cdot \frac{\partial \rm{vec}\left(\textit{V}^{\top}_{\textit{r}}\right)}{\partial \rm{vec}^{\top}\left(\hat{\textit{B}}(\lambda)\right)}.
\end{eqnarray}
(i). According to the fact that $i)$ in Proposition \ref{pro2.2}, the following results hold.
\begin{align*}
&\frac{\partial \rm{vec}\left(\textit{G}\left(\hat{\textit{B}}(\lambda)\right)\right)}{\partial \rm{vec}^{\top}\left(\textit{U}_{\textit{r}}\right)}
= \frac{\partial \rm{vec}\left(\textit{I}_{\textit{p}_{1}}\textit{U}_{\textit{r}}
\textit{V}^{\top}_{\textit{r}}\right)}
{\partial \rm{vec}^{\top}\left(\textit{U}_{\textit{r}}\right)}= \frac{\partial \left[\left(V_{r}\bigotimes I_{p_{1}}\right)
\rm{vec}\left(\textit{U}_{\textit{r}}\right)\right]}{\partial \rm{vec}^{\top}\left(\textit{U}_{\textit{r}}\right)}
= V_{r}\bigotimes I_{p_{1}}.\\
&\frac{\partial \rm{vec}\left(\textit{G}\left(\hat{\textit{B}}(\lambda)\right)\right)}{\partial \rm{vec}^{\top}\left(\textit{V}^{\top}_{\textit{r}}\right)}
= \frac{\partial \rm{vec}\left(\textit{U}_{\textit{r}}\textit{V}^{\top}_{\textit{r}}\textit{I}_{\textit{p}_{2}}\right)}
{\partial \rm{vec}^{\top}\left(\textit{V}^{\top}_{\textit{r}}\right)}= \frac{\partial \left[\left(I_{p_{2}}\bigotimes U_{r}\right) \rm{vec}\left(\textit{V}^{\top}_{\textit{r}}\right)\right]}{\partial \rm{vec}^{\top}\left(\textit{V}^{\top}_{\textit{r}}\right)}
= I_{p_{2}}\bigotimes U_{r}.
\end{align*}
(ii). According to the fact that $U_{r}$ and $V_{r}$ have full column rank, there are $U^{\top}_{r}U_{r}=I_{r}$ and $V^{\top}_{r}V_{r}=I_{r}$. Combining these facts with $W_{1}\hat{B}(\lambda)W_{2}=U_{r}\rm{Diag}\left(\textit{b}_{\textit{r}}\right)\textit{V}^{\top}_{\textit{r}}$, we have
\begin{align*}
&U_{r}=W_{1}\hat{B}(\lambda) W_{2}V_{r}\rm{Diag}\left(\textit{b}_{\textit{r}}\right)^{-1},\\
&V^{\top}_{r}=\rm{Diag}\left(\textit{b}_{\textit{r}}\right)^{-1}\textit{U}^{\top}_{\textit{r}}\textit{W}_{1}\hat{\textit{B}}(\lambda) \textit{W}_{2},
\end{align*}
which leads to
\begin{align*}
&\rm{vec}\left(\textit{U}_{\textit{r}}\right)= \left[\left(\rm{Diag}\left(\textit{b}_{\textit{r}}\right)^{-1}
\textit{V}^{\top}_{\textit{r}}\textit{W}_{2}\right)\bigotimes \textit{W}_{1}\right]\rm{vec}\left(\hat{\textit{B}}(\lambda)\right) ,\\
&\rm{vec}\left(\textit{V}^{\top}_{\textit{r}}\right)= \left[\textit{W}_{2}\bigotimes \left(\rm{Diag}\left(\textit{b}_{\textit{r}}\right)^{-1}
\textit{U}^{\top}_{\textit{r}}\textit{W}_{1}\right)\right]
\rm{vec}\left(\hat{\textit{B}}(\lambda)\right).
\end{align*}
Thus,
\begin{align*}
\frac{\partial \rm{vec}\left(\textit{U}_{\textit{r}}\right)}{\partial \rm{vec}^{\top}\left(\hat{\textit{B}}(\lambda)\right)}
=\left[\rm{Diag}\left(\textit{b}_{\textit{r}}\right)^{-1}\textit{V}^{\top}_{\textit{r}}\textit{W}_{2}\right]\bigotimes W_{1}, \frac{\partial \rm{vec}\left(\textit{V}^{\top}_{\textit{r}}\right)}{\partial \rm{vec}^{\top}\left(\hat{\textit{B}}(\lambda)\right)}=W_{2}\bigotimes \left[\rm{Diag}\left(\textit{b}_{\textit{r}}\right)^{-1}\textit{U}^{\top}_{\textit{r}}\textit{W}_{1}\right].
\end{align*}
Combining results in (i) and (ii), we know that
\begin{eqnarray*}
M^{(1)}_{r}=\left(W_{2}\bigotimes W_{1}\right)\frac{\partial \rm{vec}\left(\textit{G}\left(\hat{\textit{B}}(\lambda)\right)\right)}{\partial \rm{vec}^{\top}\left(\hat{\textit{B}}(\lambda)\right)}=
\left[W_{2}V_{r}\rm{Diag}
\left(\textit{b}_{\textit{r}}\right)^{-1}\textit{V}^{\top}_{\textit{r}}\textit{W}_{2}\right]\bigotimes W^{2}_{1}+W^{2}_{2}\bigotimes\left[W_{1}U_{r}\rm{Diag}\left(\textit{b}_{\textit{r}}\right)^{-1}
\textit{U}^{\top}_{\textit{r}}\textit{W}_{1}\right].
\end{eqnarray*}
\end{proof}
It is obvious that $W_{2}V_{r}$Diag$\left(b_{r}\right)^{-1}V^{\top}_{r}W_{2}\succeq0$, $W^{2}_{1}\succ0$, $W^{2}_{2}\succ0$ and $W_{1}U_{r}$Diag$\left(b_{r}\right)^{-1}U^{\top}_{r}W_{1}\succeq0$. Due to \cite[Corollary 2.1]{magnus2019matrix}, we have $(W_{2}V_{r}$Diag$\left(b_{r}\right)^{-1}V^{\top}_{r}W_{2})\bigotimes W^{2}_{1}\succeq0$ and $W^{2}_{2}\bigotimes
(W_{1}U_{r}$Diag$\left(b_{r}\right)^{-1}U^{\top}_{r}W_{1})\succeq0$, which leads to the matrix $M^{(1)}_{r}$ is positive semidefinite.

Now, we calculate the result of $M^{(2)}_{r}$, which is decided by $\frac{\partial \rm{vec}\left(\textit{N}\right)}{\partial \rm{vec}^{\top}\left(\hat{B}(\lambda)\right)}$.
\begin{Lemma}\label{lemma3.2}
For any $\lambda\geq0$, let $r$ be the rank of $W_{1}\hat{B}(\lambda)W_{2}$. Suppose the singular value decomposition of $W_{1}\hat{B}(\lambda)W_{2}$ is $W_{1}\hat{B}(\lambda)W_{2}=U_{r}\rm{Diag}\left(\textit{b}_{\textit{r}}\right)\textit{V}_{\textit{r}}^{\top}$ with $U_{r}\in\mathbb{R}^{p_{1}\times r}$, $V_{r}\in\mathbb{R}^{p_{2}\times r}$ and $b_{r}\in\mathbb{R}^{r}_{++}$.  Then,
\begin{equation}\label{eqnn}
\begin{aligned}
M^{(2)}_{r}=&-\left[\textit{I}_{\textit{p}_{2}}\bigotimes \left(W_{1}\left(\textit{I}_{\textit{p}_{1}}-\textit{U}_{\textit{r}}\textit{U}^{\top}_{\textit{r}}\right)W^{-1}_{1}\right)\right]
\cdot\frac{\mathcal{X}^{\top}\mathcal{X}}{n}\cdot
\left[\left(
\textit{W}^{-1}_{2}\textit{V}_{\textit{r}}\textit{V}^{\top}_{\textit{r}}\textit{W}_{2}\right)
\bigotimes I_{p_{1}}\right]\\
&-\left[\left(W_{2}\left(\textit{I}_{\textit{p}_{2}}-\textit{V}_{\textit{r}}\textit{V}^{\top}_{\textit{r}}\right)W^{-1}_{2}\right)\bigotimes
\textit{I}_{\textit{p}_{1}}\right]\cdot\frac{\mathcal{X}^{\top}\mathcal{X}}{n}\cdot \left[I_{p_{2}}\bigotimes\left(\textit{W}^{-1}_{1}\textit{U}_{\textit{r}}\textit{U}^{\top}_{\textit{r}}\textit{W}_{1}\right)\right]\\
&-\left[\left(W_{2}\left(\textit{I}_{\textit{p}_{2}}-\textit{V}_{\textit{r}}\textit{V}^{\top}_{\textit{r}}\right)W^{-1}_{2}\right)\bigotimes
\textit{I}_{\textit{p}_{1}}\right]\cdot\frac{\mathcal{X}^{\top}\mathcal{X}}{n}\cdot
\left[\left(\textit{W}^{-1}_{2}\textit{V}_{\textit{r}}\textit{V}^{\top}_{\textit{r}}\textit{W}_{2}\right)\bigotimes
\left(\textit{W}^{-1}_{1}\textit{U}_{\textit{r}}\textit{U}^{\top}_{\textit{r}}\textit{W}_{1}\right)\right]\\
&-\lambda\left[\left(\textit{E}^{\top}_{\lambda}\textit{W}^{-1}_{1}\right)
\bigotimes W_{1}\right]
\left(I_{p^{2}_{1}}+K_{p^{2}_{1}}\right)
\left[\left(\left((W_{1}\hat{B}(\lambda)W_{2})^{+}\right)^{\top}\textit{W}_{2}\right)\bigotimes \textit{W}_{1}\right]\\
&-\lambda\left[W_{2}\bigotimes \left(E_{\lambda}W^{-1}_{2}\right)\right]
\left(I_{p^{2}_{2}}+K_{p^{2}_{2}}\right)
\left[\textit{W}_{2}\bigotimes \left(\left(W_{1}\hat{B}(\lambda)W_{2}\right)^{+}\textit{W}_{1}\right)\right],
\end{aligned}
\end{equation}
where $E_{\lambda}=\frac{1}{n}\sum\limits_{i=1}^{n}\left(y_{i}-\left\langle X_{i},\hat{B}(\lambda)\right\rangle\right)X_{i}$ and $\left(W_{1}\hat{B}(\lambda)W_{2}\right)^{+}=\textit{V}_{\textit{r}}\rm{Diag}
\left(\textit{b}_{\textit{r}}\right)^{-1}\textit{U}^{\top}_{\textit{r}}$.
\end{Lemma}
\begin{proof}
Note that $M^{(2)}_{r}=\lambda\left(\textit{W}_{2}\bigotimes \textit{W}_{1}\right)\frac{\partial \rm{vec}(\textit{N})}{\partial \rm{vec}^{\top}\left(\hat{\textit{B}}(\lambda)\right)}$. Based on the optimality condition (\ref{eq}), the matrix $N$ can be expressed as
$N=\frac{1}{n\lambda}W^{-1}_{1}\left(\sum\limits_{i=1}^{n}\left(y_{i}-\left\langle X_{i},\hat{B}(\lambda)\right\rangle\right)X_{i}\right)W^{-1}_{2}-U_{r}V^{\top}_{r}$. According to $U^{\top}_{r}N=0$ and $NV_{r}=0$, the following results hold.
\begin{center}
$\frac{1}{n\lambda}U^{\top}_{r}W^{-1}_{1}\left(\sum\limits_{i=1}^{n}\left(y_{i}-\left\langle X_{i},\hat{B}(\lambda)\right\rangle\right)X_{i}\right)W^{-1}_{2}=V^{\top}_{r}$ and
$\frac{1}{n\lambda}W^{-1}_{1}\left(\sum\limits_{i=1}^{n}\left(y_{i}-\left\langle X_{i},\hat{B}(\lambda)\right\rangle\right)X_{i}\right)W^{-1}_{2}V_{r}=U_{r}$.
\end{center}
Replacing these results into $N$, then $N$ can be expressed as
\begin{equation*}
\begin{aligned}
N&=\frac{1}{n\lambda}W^{-1}_{1}\sum\limits_{i=1}^{n}\left(y_{i}-\left\langle X_{i},\hat{B}(\lambda)\right\rangle\right)X_{i}W^{-1}_{2}&\\
&\quad\quad\quad\quad-\frac{1}{n^2\lambda^2}W^{-1}_{1}\sum\limits_{i=1}^{n}\left(y_{i}-\left\langle X_{i},\hat{B}(\lambda)\right\rangle\right)X_{i}W^{-1}_{2}V_{r}U^{\top}_{r}W^{-1}_{1}
\sum\limits_{i=1}^{n}\left(y_{i}-\left\langle X_{i},\hat{B}(\lambda)\right\rangle\right)X_{i}W^{-1}_{2}\\
&=\frac{1}{n\lambda}W^{-1}_{1}\left(\sum\limits_{i=1}^{n}\left(y_{i}-\left\langle X_{i},\hat{B}(\lambda)\right\rangle\right)X_{i}\right)W^{-1}_{2}
\left[I_{p_{2}}-\frac{1}{n\lambda}V_{r}U^{\top}_{r}W^{-1}_{1}\left(\sum\limits_{i=1}^{n}
\left(y_{i}-\left\langle X_{i},\hat{B}(\lambda)\right\rangle\right)X_{i}\right)W^{-1}_{2}\right]\\
&=\left[I_{p_{1}}-\frac{1}{n\lambda}W^{-1}_{1}
\left(\sum\limits_{i=1}^{n}\left(y_{i}-\left\langle X_{i},\hat{B}(\lambda)\right\rangle\right)X_{i}\right)W^{-1}_{2}V_{r}U^{\top}_{r}\right]
\frac{1}{n\lambda}W^{-1}_{1}\left(\sum\limits_{i=1}^{n}\left(y_{i}-\left\langle X_{i},\hat{B}(\lambda)\right\rangle\right)X_{i}\right)W^{-1}_{2}.
\end{aligned}
\end{equation*}
Let $\tilde{\textit{N}}=\frac{1}{n\lambda}W^{-1}_{1}\left(\sum\limits_{i=1}^{n}\left(y_{i}-\left\langle X_{i},\hat{B}(\lambda)\right\rangle\right)X_{i}\right)W^{-1}_{2}$. Then,
\begin{equation}\label{eqN}
\begin{aligned}
N&=\frac{1}{n\lambda}W^{-1}_{1}\left(\sum\limits_{i=1}^{n}\left(y_{i}-\left\langle X_{i},\hat{B}(\lambda)\right\rangle\right)X_{i}\right)W^{-1}_{2}
\left[I_{p_{2}}-V_{r}V^{\top}_{r}\right]=\tilde{\textit{N}}\left[I_{p_{2}}-V_{r}V^{\top}_{r}\right]\\
&=\frac{1}{n\lambda}\left[I_{p_{1}}-U_{r}U^{\top}_{r}\right]W^{-1}_{1}\left(\sum\limits_{i=1}^{n}\left(y_{i}-\left\langle X_{i},\hat{B}(\lambda)\right\rangle\right)X_{i}\right)W^{-1}_{2}
=\left[I_{p_{1}}-U_{r}U^{\top}_{r}\right]\tilde{\textit{N}}.
\end{aligned}
\end{equation}
To verify this $N$ satisfies the required conditions, there is only one left part needed to be proved, i.e, $\|N\|_{2}\leq 1$. According to the equation (\ref{eq}),
\begin{eqnarray*}
\tilde{\textit{N}}=\frac{1}{n\lambda}W^{-1}_{1}\left(\sum\limits_{i=1}^{n}\left(y_{i}-\left\langle X_{i},\hat{B}(\lambda)\right\rangle\right)X_{i}\right)W^{-1}_{2}=\partial \left\|\hat{B}(\lambda)\right\|_{*},
\end{eqnarray*}
which leads to $\left\|\tilde{\textit{N}}\right\|_{2}\leq 1$. Based on the singular inequalities in \cite{horn2012matrix},
\begin{eqnarray*}
\left\|\textit{N}\right\|_{2}
=\left\|\tilde{\textit{N}}\left[I_{p_{2}}-V_{r}V^{\top}_{r}\right]\right\|_{2}\leq
\left\|\tilde{\textit{N}}\right\|_{2}\cdot\left\|I_{p_{2}}-V_{r}V^{\top}_{r}\right\|_{2}\leq 1.
\end{eqnarray*}
So, the matrix $N$ in (\ref{eq}) can be expressed as in (\ref{eqN}).

Now, we calculate the result of $\frac{\partial \rm{vec}(\textit{N})}{\partial \rm{vec}^{\top}\left(\hat{B}(\lambda)\right)}$. It is sure that
\begin{eqnarray*}
\frac{\partial \rm{vec}(\textit{N})}{\partial \rm{vec}^{\top}\left(\hat{B}(\lambda)\right)}
=\frac{\partial \rm{vec}(\textit{N})}{\partial \rm{vec}^{\top}\left(\textit{U}_{\textit{r}}\right)}\cdot\frac{\partial \rm{vec}\left(\textit{U}_{\textit{r}}\right)}{\partial \rm{vec}^{\top}\left(\hat{B}(\lambda)\right)}+\frac{\partial \rm{vec}(\textit{N})}{\partial \rm{vec}^{\top}\left(\textit{V}^{\top}_{\textit{r}}\right)}\cdot\frac{\partial \rm{vec}\left(\textit{V}^{\top}_{\textit{r}}\right)}{\partial \rm{vec}^{\top}\left(\hat{B}(\lambda)\right)}+\frac{\partial \rm{vec}(\textit{N})}{\partial \rm{vec}^{\top}\left(\rm{Diag}\left(\textit{b}_{\textit{r}}\right)\right)}\cdot\frac{\partial \rm{vec}\left(\rm{Diag}\left(\textit{b}_{\textit{r}}\right)\right)}{\partial \rm{vec}^{\top}\left(\hat{B}(\lambda)\right)}.
\end{eqnarray*}
(a). Based on the expression of $N$,
$\frac{\partial \rm{vec}(\textit{N})}{\partial \rm{vec}^{\top}\left(\textit{U}_{\textit{r}}\right)}=\frac{\partial \rm{vec}(\textit{N})}{\partial \rm{vec}^{\top}\left(\tilde{\textit{N}}\right)}\cdot \frac{\partial \rm{vec}\left(\tilde{\textit{N}}\right)}{\partial \rm{vec}^{\top}\left(\textit{U}_{\textit{r}}\right)}+\frac{\partial \rm{vec}(\textit{N})}{\partial \rm{vec}^{\top}\left(\textit{I}_{\textit{p}_{1}}
-\textit{U}_{\textit{r}}\textit{U}^{\top}_{\textit{r}}\right)}\cdot \frac{\partial \rm{vec}
\left(\textit{I}_{\textit{p}_{1}}-\textit{U}_{\textit{r}}\textit{U}^{\top}_{\textit{r}}\right)}{\partial \rm{vec}^{\top}\left(\textit{U}_{\textit{r}}\right)}.$
\begin{eqnarray*}
\frac{\partial \rm{vec}(\textit{N})}{\partial \rm{vec}^{\top}\left(\tilde{\textit{N}}\right)}=
\frac{\partial \rm{vec}\left(\left(\textit{I}_{\textit{p}_{2}}\bigotimes\left(\textit{I}_{\textit{p}_{1}}
-\textit{U}_{\textit{r}}\textit{U}^{\top}_{\textit{r}}\right)\right)
\rm{vec}\left(\tilde{\textit{N}}\right)\right)}{\partial \rm{vec}^{\top}\left(\tilde{\textit{N}}\right)}
=\textit{I}_{\textit{p}_{2}}\bigotimes
\left(\textit{I}_{\textit{p}_{1}}-\textit{U}_{\textit{r}}\textit{U}^{\top}_{\textit{r}}\right).
\end{eqnarray*}
Due to the fact that
\begin{equation*}
\begin{aligned}
\rm{vec}(\tilde{\textit{N}})&=\frac{1}{\textit{n}\lambda}
\left(W^{-1}_{2}\bigotimes W^{-1}_{1}\right) \rm{vec}\left(\sum\limits_{\textit{i}=1}^{\textit{n}}\left(\textit{y}_{\textit{i}}-\left\langle \textit{X}_{\textit{i}},\hat{B}(\lambda)\right\rangle\right)\textit{X}_{\textit{i}}\right)\\
&=\frac{1}{\textit{n}\lambda}\left(W^{-1}_{2}\bigotimes W^{-1}_{1}\right)\left[\mathcal{X}^{\top}\textbf{y}-
\mathcal{X}^{\top}\mathcal{X}\rm{vec}\left(\hat{B}(\lambda)\right)\right]\\
&=\frac{1}{\textit{n}\lambda}\left(W^{-1}_{2}\bigotimes W^{-1}_{1}\right)\left[\mathcal{X}^{\top}\textbf{y}
-\mathcal{X}^{\top}\mathcal{X}\left(\left(\textit{W}^{-1}_{2}\textit{V}_{\textit{r}}\rm{Diag}\left(\textit{b}_{\textit{r}}\right)\right)\bigotimes \textit{W}^{-1}_{1}\right)\rm{vec}\left(\textit{U}_{\textit{r}}\right)\right],
\end{aligned}
\end{equation*}
the result of $\frac{\partial \rm{vec}(\tilde{\textit{N}})}{\partial \rm{vec}^{\top}\left(\textit{U}_{\textit{r}}\right)}$ is
\begin{eqnarray*}
\frac{\partial \rm{vec}\left(\tilde{\textit{N}}\right)}{\partial \rm{vec}^{\top}\left(\textit{U}_{\textit{r}}\right)}
=-\frac{1}{n\lambda}\left(W^{-1}_{2}\bigotimes W^{-1}_{1}\right)\mathcal{X}^{\top}\mathcal{X}\left[\left(\textit{W}^{-1}_{2}
\textit{V}_{\textit{r}}\rm{Diag}\left(\textit{b}_{\textit{r}}\right)\right)\bigotimes \textit{W}^{-1}_{1}\right].
\end{eqnarray*}
It is easy to get that
\begin{eqnarray*}
\frac{\partial \rm{vec}(\textit{N})}{\partial \rm{vec}^{\top}\left(\textit{I}_{\textit{p}_{1}}-\textit{U}_{\textit{r}}\textit{U}^{\top}_{\textit{r}}\right)}
=\frac{\partial \left(\left(\tilde{\textit{N}}^{\top}\bigotimes \textit{I}_{\textit{p}_{1}}\right)\rm{vec}\left(\textit{I}_{\textit{p}_{\textit{1}}}-\textit{U}_{\textit{r}}
\textit{U}^{\top}_{\textit{r}}\right)\right)}{\partial \rm{vec}^{\top}\left(\textit{I}_{\textit{p}_{1}}-\textit{U}_{\textit{r}}\textit{U}^{\top}_{\textit{r}}\right)}
=\tilde{\textit{N}}^{\top}\bigotimes \textit{I}_{\textit{p}_{1}}.
\end{eqnarray*}
The claim $iii)$ in Proposition \ref{pro2.2} leads to vec$\left(U^{\top}_{r}\right)=K_{rp_{1}}$vec$\left(U_{r}\right)$ and
\begin{equation*}
\begin{aligned}
\frac{\partial \rm{vec}\left(\textit{I}_{\textit{p}_{1}}-\textit{U}_{\textit{r}}\textit{U}^{\top}_{\textit{r}}\right)}{\partial \rm{vec}^{\top}\left(\textit{U}_{\textit{r}}\right)}&=
-\frac{\partial \rm{vec}\left(\textit{U}_{\textit{r}}\textit{U}^{\top}_{\textit{r}}\right)}{\partial \rm{vec}^{\top}\left(\textit{U}_{\textit{r}}\right)}-\frac{\partial \rm{vec}\left(\textit{U}_{\textit{r}}\textit{U}^{\top}_{\textit{r}}\right)}{\partial \rm{vec}^{\top}\left(\textit{U}^{\top}_{\textit{r}}\right)}\cdot \frac{\partial \rm{vec}(\textit{U}^{\top}_{\textit{r}})}{\partial \rm{vec}^{\top}\left(\textit{U}_{\textit{r}}\right)}\\
&=-\textit{U}_{\textit{r}}\bigotimes I_{p_{1}}-\left(I_{p_{1}}\bigotimes U_{r}\right)K_{rp_{1}}\\
&=-\textit{U}_{\textit{r}}\bigotimes I_{p_{1}}-K_{p^{2}_{1}}\left(\textit{U}_{\textit{r}}\bigotimes I_{p_{1}}\right)=-\left(I_{p^{2}_{1}}+K_{p^{2}_{1}}\right)\left(\textit{U}_{\textit{r}}\bigotimes I_{p_{1}}\right).
\end{aligned}
\end{equation*}
Hence,
\begin{equation*}
\begin{aligned}
\frac{\partial \rm{vec}(\textit{N})}{\partial \rm{vec}^{\top}(\textit{U}_{\textit{r}})}
=&-\frac{1}{n\lambda}\left[W^{-1}_{2}\bigotimes (\textit{I}_{\textit{p}_{1}}-\textit{U}_{\textit{r}}\textit{U}^{\top}_{\textit{r}})W^{-1}_{1}\right]
\mathcal{X}^{\top}\mathcal{X}\left[\left(\textit{W}^{-1}_{2}\textit{V}_{\textit{r}}\rm{Diag}\left(\textit{b}_{\textit{r}}\right)\right)\bigotimes \textit{W}^{-1}_{1}\right]\\
&-\left(\tilde{\textit{N}}^{\top}\bigotimes \textit{I}_{\textit{p}_{1}}\right)\left(I_{p^{2}_{1}}+K_{p^{2}_{1}}\right)
\left(\textit{U}_{\textit{r}}\bigotimes I_{p_{1}}\right).
\end{aligned}
\end{equation*}
(b). Based on the expression of $N$,
$\frac{\partial \rm{vec}(\textit{N})}{\partial \rm{vec}^{\top}\left(\textit{V}^{\top}_{\textit{r}}\right)}=\frac{\partial \rm{vec}(\textit{N})}{\partial \rm{vec}^{\top}\left(\tilde{\textit{N}}\right)}\cdot \frac{\partial \rm{vec}\left(\tilde{\textit{N}}\right)}{\partial \rm{vec}^{\top}\left(\textit{V}^{\top}_{\textit{r}}\right)}+\frac{\partial \rm{vec}(\textit{N})}{\partial \rm{vec}^{\top}\left(\textit{I}_{\textit{p}_{2}}-\textit{V}_{\textit{r}}\textit{V}^{\top}_{\textit{r}}\right)}\cdot \frac{\partial \rm{vec}\left(\textit{I}_{\textit{p}_{2}}-\textit{V}_{\textit{r}}\textit{V}^{\top}_{\textit{r}}\right)}{\partial \rm{vec}^{\top}\left(\textit{V}^{\top}_{\textit{r}}\right)}.$
\begin{eqnarray*}
\frac{\partial \rm{vec}(\textit{N})}{\partial \rm{vec}^{\top}\left(\tilde{\textit{N}}\right)}=
\frac{\partial \rm{vec}\left(\left(\left(\textit{I}_{\textit{p}_{2}}-\textit{V}_{\textit{r}}\textit{V}^{\top}_{\textit{r}}\right)\bigotimes\textit{I}_{\textit{p}_{1}}\right)
\rm{vec}(\tilde{\textit{N}})\right)}{\partial \rm{vec}^{\top}\left(\tilde{\textit{N}}\right)}=\left(\textit{I}_{\textit{p}_{2}}-\textit{V}_{\textit{r}}\textit{V}^{\top}_{\textit{r}}\right)\bigotimes\textit{I}_{\textit{p}_{1}}.
\end{eqnarray*}
Due to the fact that
\begin{equation*}
\begin{aligned}
\rm{vec}\left(\tilde{\textit{N}}\right)&=\frac{1}{\textit{n}\lambda}
\left(W^{-1}_{2}\bigotimes W^{-1}_{1}\right) \rm{vec}\left(\sum\limits_{\textit{i}=1}^{\textit{n}}\left(\textit{y}_{\textit{i}}-\left\langle \textit{X}_{\textit{i}},\hat{\textit{B}}(\lambda)\right\rangle\right)\textit{X}_{\textit{i}}\right)\\
&=\frac{1}{\textit{n}\lambda}\left(W^{-1}_{2}\bigotimes W^{-1}_{1}\right)
\left[\mathcal{X}^{\top}\textbf{y}-\mathcal{X}^{\top}\mathcal{X}\rm{vec}(\hat{\textit{B}}(\lambda))\right]\\
&=\frac{1}{\textit{n}\lambda}\left(W^{-1}_{2}\bigotimes W^{-1}_{1}\right)
\left[\mathcal{X}^{\top}\textbf{y}-\mathcal{X}^{\top}\mathcal{X}\left(\textit{W}^{-1}_{2}\bigotimes \left(\textit{W}^{-1}_{1}\textit{U}_{\textit{r}}\rm{Diag}\left(\textit{b}_{\textit{r}}\right)\right)\right)\rm{vec}\left(\textit{V}^{\top}_{\textit{r}}\right)\right],
\end{aligned}
\end{equation*}
the result of $\frac{\partial \rm{vec}\left(\tilde{\textit{N}}\right)}{\partial \rm{vec}^{\top}\left(\textit{V}^{\top}_{\textit{r}}\right)}$ is
\begin{eqnarray*}
\frac{\partial \rm{vec}\left(\tilde{\textit{N}}\right)}{\partial \rm{vec}^{\top}\left(\textit{V}^{\top}_{\textit{r}}\right)}
=-\frac{1}{n\lambda}\left(W^{-1}_{2}\bigotimes W^{-1}_{1}\right)
\mathcal{X}^{\top}\mathcal{X}\left[\textit{W}^{-1}_{2}\bigotimes \left(\textit{W}^{-1}_{1}\textit{U}_{\textit{r}}\rm{Diag}\left(\textit{b}_{\textit{r}}\right)\right)\right].
\end{eqnarray*}
It is easy to get that
\begin{eqnarray*}
\frac{\partial \rm{vec}(\textit{N})}{\partial \rm{vec}^{\top}\left(\textit{I}_{\textit{p}_{2}}-\textit{V}_{\textit{r}}\textit{V}^{\top}_{\textit{r}}\right)}
=\frac{\partial \rm{vec}\left(\left(\textit{I}_{\textit{p}_{2}}\bigotimes \tilde{\textit{N}}\right)\rm{vec}\left(\textit{I}_{\textit{p}_{2}}-\textit{V}_{\textit{r}}\textit{V}^{\top}_{\textit{r}}\right)\right)}{\partial \rm{vec}^{\top}\left(\textit{I}_{\textit{p}_{2}}-\textit{V}_{\textit{r}}\textit{V}^{\top}_{\textit{r}}\right)}=\textit{I}_{\textit{p}_{2}}\bigotimes \tilde{\textit{N}}.
\end{eqnarray*}
From the result that vec$\left(V^{\top}\right)=K_{rp_{2}}$vec$\left(V_{r}\right)$, we have
\begin{equation*}
\begin{aligned}
\frac{\partial \rm{vec}\left(\textit{I}_{\textit{p}_{2}}-\textit{V}_{\textit{r}}\textit{V}^{\top}_{\textit{r}}\right)}{\partial \rm{vec}^{\top}\left(\textit{V}^{\top}_{\textit{r}}\right)}&=
-\frac{\partial \rm{vec}\left(\textit{V}_{\textit{r}}\textit{V}^{\top}_{\textit{r}}\right)}{\partial \rm{vec}^{\top}\left(\textit{V}^{\top}_{\textit{r}}\right)}-\frac{\partial \rm{vec}\left(\textit{V}_{\textit{r}}\textit{V}^{\top}_{\textit{r}}\right)}{\partial \rm{vec}^{\top}\left(\textit{V}_{\textit{r}}\right)}\cdot \frac{\partial \rm{vec}\left(\textit{V}_{\textit{r}}\right)}{\partial \rm{vec}^{\top}\left(\textit{V}^{\top}_{\textit{r}}\right)}\\
&=-I_{p_{2}}\bigotimes\textit{V}_{\textit{r}} -\left(V_{r}\bigotimes I_{p_{1}}\right)K_{rp_{2}}\\
&=-I_{p_{2}}\bigotimes\textit{V}_{\textit{r}}-K_{p^{2}_{2}}
\left(I_{p_{2}}\bigotimes\textit{V}_{\textit{r}}\right)
=-\left(I_{p^{2}_{2}}+K_{p^{2}_{2}}
\right)\left(I_{p_{2}}\bigotimes\textit{V}_{\textit{r}}\right).
\end{aligned}
\end{equation*}
Hence,
\begin{equation*}
\begin{aligned}
\frac{\partial \rm{vec}(\textit{N})}{\partial \rm{vec}^{\top}\left(\textit{V}^{\top}_{\textit{r}}\right)}
=&-\frac{1}{n\lambda}\left[\left(\left(\textit{I}_{\textit{p}_{2}}-\textit{V}_{\textit{r}}
\textit{V}^{\top}_{\textit{r}}\right)W^{-1}_{2}\right)\bigotimes W^{-1}_{1}\right]
\mathcal{X}^{\top}\mathcal{X}\left[\textit{W}^{-1}_{2}\bigotimes \left(\textit{W}^{-1}_{1}\textit{U}_{\textit{r}}\rm{Diag}\left(\textit{b}_{\textit{r}}\right)\right)\right]\\
&-\left(\textit{I}_{\textit{p}_{2}}\bigotimes \tilde{\textit{N}}\right)\left(I_{p^{2}_{2}}+K_{p^{2}_{2}}\right)
\left(I_{p_{2}}\bigotimes\textit{V}_{\textit{r}}\right).
\end{aligned}
\end{equation*}
(c). In the similar way, the result of $\frac{\partial \rm{vec}(\textit{N})}{\partial \rm{vec}^{\top}\left(\rm{Diag}\left(\textit{b}_{\textit{r}}\right)\right)}$ is showed as follows.
\begin{equation*}
\begin{aligned}
\frac{\partial \rm{vec}(\textit{N})}{\partial \rm{vec}^{\top}\left(\rm{Diag}\left(\textit{b}_{\textit{r}}\right)\right)}
&=\frac{\partial \rm{vec}(\textit{N})}{\partial \rm{vec}^{\top}\left(\tilde{\textit{N}}\right)}\cdot
\frac{\partial \rm{vec}\left(\tilde{\textit{N}}\right)}{\partial \rm{vec}^{\top}\left(\rm{Diag}\left(\textit{b}_{\textit{r}}\right)\right)}\\
&=-\frac{1}{n\lambda}\left[\left(\textit{I}_{\textit{p}_{2}}
-\textit{V}_{\textit{r}}\textit{V}^{\top}_{\textit{r}}\right)\bigotimes\textit{I}_{\textit{p}_{1}}\right]
\cdot \left(\textit{W}^{-1}_{2}\bigotimes \textit{W}^{-1}_{1}\right)\mathcal{X}^{\top}\mathcal{X}
\left[\left(\textit{W}^{-1}_{2}V_{r}\right)\bigotimes\left(\textit{W}^{-1}_{1}U_{r}\right)\right]\\
&=-\frac{1}{n\lambda}\left[\left(\left(\textit{I}_{\textit{p}_{2}}
-\textit{V}_{\textit{r}}\textit{V}^{\top}_{\textit{r}}\right)\textit{W}^{-1}_{2}\right)\bigotimes \textit{W}^{-1}_{1})\right]\mathcal{X}^{\top}\mathcal{X}
\left[\left(\textit{W}^{-1}_{2}V_{r}\right)\bigotimes\left(\textit{W}^{-1}_{1}U_{r}\right)\right].
\end{aligned}
\end{equation*}
(d). Because $W_{1}\hat{B}(\lambda)W_{2}=U_{r}$Diag$(b_{r})V^{\top}_{r}$, then Diag$(b_{r})=U^{\top}_{r}W_{1}\hat{B}(\lambda)W_{2}V_{r}$ and
\begin{eqnarray*}
\frac{\partial\rm{vec}\left(Diag\left(\textit{b}_{\textit{r}}\right)\right)
}{\partial\rm{vec}^{\top}\left(\hat{B}(\lambda)\right)}
=\left(V^{\top}_{r}W_{2}\right)\bigotimes\left(U^{\top}_{r}W_{1}\right).
\end{eqnarray*}
In addition, results of $\frac{\partial \rm{vec}\left(\textit{U}_{\textit{r}}\right)}{\partial \rm{vec}^{\top}\left(\hat{\textit{B}}(\lambda)\right)}$ and $\frac{\partial \rm{vec}\left(\textit{V}^{\top}_{\textit{r}}\right)}{\partial \rm{vec}^{\top}\left(\hat{\textit{B}}(\lambda)\right)}$ have been computed in Lemma \ref{lemma3.1}.

From (a)-(d), we know that
%\begin{equation*}
%\begin{aligned}
%\frac{\partial \rm{vec}(\textit{N})}{\partial \rm{vec}^{\textit{T}}(\hat{B}(\lambda))}
%=&-\frac{1}{n\lambda}\left[\left((\textit{I}_{\textit{p}_{2}}-\textit{V}_{\textit{r}}\textit{V}^{\textit{T}}_{\textit{r}})W^{-1}_{2}\right)\bigotimes ((\textit{I}_{\textit{p}_{1}}-\textit{U}_{\textit{r}}\textit{U}^{\textit{T}}_{\textit{r}})W^{-1}_{1})\right]
%\mathcal{X}^{\top}\mathcal{X}\cdot\\
%&~~~~\left[(\textit{W}^{-1}_{2}\textit{V}_{\textit{r}}\textit{V}^{\textit{T}}_{\textit{r}}\textit{W}_{2})\bigotimes I_{p_{1}}+ I_{p_{2}}\bigotimes (\textit{W}^{-1}_{1}\textit{U}_{\textit{r}}\textit{U}^{\textit{T}}_{\textit{r}}\textit{W}_{1})
%+(\textit{W}^{-1}_{2}\textit{V}_{\textit{r}}\textit{V}^{\textit{T}}_{\textit{r}}\textit{W}_{2})\bigotimes (\textit{W}^{-1}_{1}\textit{U}_{\textit{r}}\textit{U}^{\textit{T}}_{\textit{r}}\textit{W}_{1})\right]\\
%&-(\tilde{\textit{N}}^{\top}\bigotimes \textit{I}_{\textit{p}_{1}})(I_{p^{2}_{1}}+K_{p^{2}_{1}})
%\left[(\textit{U}_{\textit{r}}\rm{Diag}(\textit{b}_{\textit{r}})^{-1}\textit{V}^{\textit{T}}_{\textit{r}}\textit{W}_{2})\bigotimes \textit{W}_{1}\right]\\
%&-(\textit{I}_{\textit{p}_{2}}\bigotimes \tilde{\textit{N}})(I_{p^{2}_{2}}+K_{p^{2}_{2}})
%\left[\textit{W}_{2}\bigotimes (\textit{V}_{\textit{r}}\rm{Diag}(\textit{b}_{\textit{r}})^{-1}\textit{U}^{\textit{T}}_{\textit{r}}\textit{W}_{1})\right].
%\end{aligned}
%\end{equation*}
%Then,
\begin{equation*}
\begin{aligned}
M^{(2)}_{r}=&\lambda \left(W_{2}\bigotimes W_{1}\right)\cdot \frac{\partial \rm{vec}(\textit{N})}{\partial \rm{vec}^{\top}\left(\hat{B}(\lambda)\right)}\\
&=-\left[\textit{I}_{\textit{p}_{2}}\bigotimes \left(W_{1}\left(\textit{I}_{\textit{p}_{1}}-\textit{U}_{\textit{r}}\textit{U}^{\top}_{\textit{r}}\right)W^{-1}_{1}\right)\right]
\cdot\frac{\mathcal{X}^{\top}\mathcal{X}}{n}\cdot
\left[\left(\textit{W}^{-1}_{2}\textit{V}_{\textit{r}}\textit{V}^{\top}_{\textit{r}}\textit{W}_{2}\right)\bigotimes I_{p_{1}}\right]\\
&~~~~-\left[\left(W_{2}(\textit{I}_{\textit{p}_{2}}-\textit{V}_{\textit{r}}\textit{V}^{\top}_{\textit{r}})W^{-1}_{2}\right)\bigotimes
\textit{I}_{\textit{p}_{1}}\right]\cdot\frac{\mathcal{X}^{\top}\mathcal{X}}{n}\cdot \left[I_{p_{2}}\bigotimes\left(\textit{W}^{-1}_{1}\textit{U}_{\textit{r}}\textit{U}^{\top}_{\textit{r}}\textit{W}_{1}\right)\right]\\
&~~~~-\left[\left(W_{2}\left(\textit{I}_{\textit{p}_{2}}-\textit{V}_{\textit{r}}\textit{V}^{\top}_{\textit{r}}\right)W^{-1}_{2}\right)\bigotimes
\textit{I}_{\textit{p}_{1}}\right]\cdot\frac{\mathcal{X}^{\top}\mathcal{X}}{n}\cdot
\left[\left(\textit{W}^{-1}_{2}\textit{V}_{\textit{r}}\textit{V}^{\top}_{\textit{r}}\textit{W}_{2}\right)\bigotimes
\left(\textit{W}^{-1}_{1}\textit{U}_{\textit{r}}\textit{U}^{\top}_{\textit{r}}\textit{W}_{1}\right)\right]\\
&~~~~-\lambda\left[\left(W_{2}\tilde{\textit{N}}^{\top}\right)\bigotimes W_{1}\right]\left(I_{p^{2}_{1}}+K_{p^{2}_{1}}\right)
\left[\left(\textit{U}_{\textit{r}}\rm{Diag}(\textit{b}_{\textit{r}})^{-1}\textit{V}^{\top}_{\textit{r}}\textit{W}_{2}\right)\bigotimes \textit{W}_{1}\right]\\
&~~~~-\lambda\left[W_{2}\bigotimes \left(W_{1}\tilde{\textit{N}}\right)\right]\left(I_{p^{2}_{2}}+K_{p^{2}_{2}}\right)
\left[\textit{W}_{2}\bigotimes \left(\textit{V}_{\textit{r}}\rm{Diag}(\textit{b}_{\textit{r}})^{-1}\textit{U}^{\top}_{\textit{r}}\textit{W}_{1}\right)\right].
\end{aligned}
\end{equation*}
By simple computation, we know that $\left(W_{1}\hat{B}(\lambda)W_{2}\right)^{+}=\textit{V}_{\textit{r}}\rm{Diag}
\left(\textit{b}_{\textit{r}}\right)^{-1}\textit{U}^{\top}_{\textit{r}}$. Replace the expression of $\tilde{\textit{N}}$ and $E_{\lambda}=\frac{1}{n}\sum\limits_{i=1}^{n}\left(y_{i}-\left\langle X_{i},\hat{B}(\lambda)\right\rangle\right)X_{i}$ into $M^{(2)}_{r}$.  Then, we have
\begin{equation*}
\begin{aligned}
M^{(2)}_{r}=&-\left[\textit{I}_{\textit{p}_{2}}\bigotimes \left(W_{1}\left(\textit{I}_{\textit{p}_{1}}-\textit{U}_{\textit{r}}\textit{U}^{\top}_{\textit{r}}\right)W^{-1}_{1}\right)\right]
\cdot\frac{\mathcal{X}^{\top}\mathcal{X}}{n}\cdot
\left[\left(\textit{W}^{-1}_{2}\textit{V}_{\textit{r}}\textit{V}^{\top}_{\textit{r}}\textit{W}_{2}\right)\bigotimes I_{p_{1}}\right]\\
&-\left[\left(W_{2}\left(\textit{I}_{\textit{p}_{2}}-\textit{V}_{\textit{r}}\textit{V}^{\top}_{\textit{r}}\right)W^{-1}_{2}\right)\bigotimes
\textit{I}_{\textit{p}_{1}}\right]\cdot\frac{\mathcal{X}^{\top}\mathcal{X}}{n}\cdot \left[I_{p_{2}}\bigotimes\left(\textit{W}^{-1}_{1}\textit{U}_{\textit{r}}\textit{U}^{\top}_{\textit{r}}\textit{W}_{1}\right)\right]\\
&-\left[\left(W_{2}(\textit{I}_{\textit{p}_{2}}-\textit{V}_{\textit{r}}\textit{V}^{\top}_{\textit{r}})W^{-1}_{2}\right)\bigotimes
\textit{I}_{\textit{p}_{1}}\right]\cdot\frac{\mathcal{X}^{\top}\mathcal{X}}{n}\cdot
\left[\left(\textit{W}^{-1}_{2}\textit{V}_{\textit{r}}\textit{V}^{\top}_{\textit{r}}\textit{W}_{2}\right)\bigotimes
\left(\textit{W}^{-1}_{1}\textit{U}_{\textit{r}}\textit{U}^{\top}_{\textit{r}}\textit{W}_{1}\right)\right]\\
&-\lambda\left[\left(\textit{E}^{\top}_{\lambda}\textit{W}^{-1}_{1}\right)
\bigotimes W_{1}\right]
\left(I_{p^{2}_{1}}+K_{p^{2}_{1}}\right)
\left[\left(\left(\left(W_{1}\hat{B}(\lambda)W_{2}\right)^{+}\right)^{\top}
\textit{W}_{2}\right)\bigotimes \textit{W}_{1}\right]\\
&-\lambda\left[W_{2}\bigotimes \left(E_{\lambda}W^{-1}_{2}\right)\right]
\left(I_{p^{2}_{2}}+K_{p^{2}_{2}}\right)
\left[\textit{W}_{2}\bigotimes \left(\left(W_{1}\hat{B}(\lambda)W_{2}\right)^{+}\textit{W}_{1}\right)\right].
\end{aligned}
\end{equation*}
\end{proof}
\begin{Remark}\label{remark 3.2}
There are some facts about the result of this degrees of freedom.

In (\ref{eq1}), if $\lambda=0$ and $B_{\rm LS}$ has full rank, we have $M^{(2)}_{r}=0$ and
 $M_{r}=\frac{1}{\textit{n}}\sum\limits_{\textit{i}=1}^{\textit{n}}
\rm{vec}\left(\textit{X}_{\textit{i}}\right)\rm{vec}\left(\textit{X}_{\textit{i}}\right)^{\top}$, which leads to
\begin{equation*}
\begin{aligned}
\hat{\rm{df}}_{\lambda}&=\frac{1}{\textit{n}}\sum\limits_{\textit{k}=1}^{\textit{n}}
\rm{vec}(\textit{X}_{\textit{k}})^{\top}\left[\frac{1}{\textit{n}}\sum\limits_{\textit{i}=1}^{\textit{n}}
\rm{vec}\left(\textit{X}_{\textit{i}}\right)\rm{vec}\left(\textit{X}_{\textit{i}}\right)^{\top}\right]^{+}\rm{vec}(\textit{X}_{\textit{k}})=\rm{rank}\left(\frac{1}{\textit{n}}\sum\limits_{\textit{i}=1}^{\textit{n}}
\rm{vec}\left(\textit{X}_{\textit{i}}\right)\rm{vec}\left(\textit{X}_{\textit{i}}\right)^{\top}\right).
\end{aligned}
\end{equation*}

When $\lambda$ is sufficiently large such that $\hat{B}(\lambda)=0$, we have $M^{(1)}_{r}=0$ and $M^{(2)}_{r}=0$, which also leads to
\begin{center}
$\hat{\rm{df}}_{\lambda}=\rm{rank}\left(\frac{1}{\textit{n}}\sum\limits_{\textit{i}=1}^{\textit{n}}
\rm{vec}\left(\textit{X}_{\textit{i}}\right)\rm{vec}\left(\textit{X}_{\textit{i}}\right)^{\top}\right).$
\end{center}

In Zhou and Li \cite{zhou2014regularized}, they considered the special case of (\ref{eq1}), with the weight matrixes $W_{1}=I_{p_{1}}$ and $W_{2}=I_{p_{2}}$. In \cite{zhou2014regularized}, they presented the value of degrees of freedom of their model, under the assumption that $\sum\limits_{\textit{i}=1}^{\textit{n}}
\rm{vec}\left(\textit{X}_{\textit{i}}\right)\rm{vec}\left(\textit{X}_{\textit{i}}\right)^{\top}=I_{p_{1}p_{2}}$. Under this condition and $\lambda=0$, our degrees of freedom is
$\hat{\rm{df}}_{0}=p_{1}p_{2}$, which equals to the degrees of freedom in their paper.
\end{Remark}
%\begin{figure}[htbp]
%\begin{minipage}[t]{0.4\linewidth}
%\centering
%\includegraphics[width=2.5in]{df_1.eps}
%%\caption{$n=100,p_{1}=8,p_{2}=6$}
%\end{minipage}
%\begin{minipage}[t]{0.4\linewidth}
%\centering
%\includegraphics[width=2.5in]{df_4.eps}
%\end{minipage}
%\begin{minipage}[t]{0.4\linewidth}
%\centering
%\includegraphics[width=2.5in]{df_2.eps}
%\end{minipage}
%\begin{minipage}[t]{0.4\linewidth}
%\centering
%\includegraphics[width=2.5in]{df_3.eps}
%\end{minipage}
%\caption{The relationship between data sets and degrees of freedom.}
%\end{figure}
%From Figure 1, we know the following results.
%\begin{itemize}
%\item{No matter $p_{1}p_{2}>n$ or $p_{1}p_{2}<n$, $\hat{\rm{df}}_{\lambda}\leq n$ holds.}
%\item{If $p_{1}p_{2}<n$, $\hat{\rm{df}}_{\lambda}=n$ holds only when rank$(\hat{B}(\lambda))=0$.}
%\item{If $p_{1}p_{2}>n$, $\hat{\rm{df}}_{\lambda}=n$ holds when rank$(\hat{B}(\lambda))$ is small.}
%\item{If $p_{1}p_{2}\gg n$, $\hat{\rm{df}}_{\lambda}=n$ may hold for  any rank of $(\hat{B}(\lambda))$.}
%\end{itemize}
%\newpage

For any $\lambda\geq0$, we define a Bayesian information criterion (BIC) as
\begin{eqnarray}\label{eq2}
\rm{BIC}_{\lambda}=\rm{log}\left(\frac{1}{\textit{n}}
\sum\limits_{\textit{i}=1}^{\textit{n}}\left(\textit{y}_{\textit{i}}-\left\langle \textit{X}_{\textit{i}},\hat{\textit{B}}(\lambda)\right\rangle\right)^{2}\right)+\hat{\rm{df}}_{\lambda}\cdot\frac{\rm{log}(\textit{n})}{\textit{n}},
\end{eqnarray}
where $\hat{B}(\lambda)$ is the solution of (\ref{eq1}) and $\hat{\rm{df}}_{\lambda}$  is given in Theorem \ref{lemma df}. Then, the optimal tuning parameter is
\begin{center}
$\lambda^{*}=\underset{\lambda\geq0}{\arg\min}$BIC$_{\lambda}=
\underset{\lambda\geq0}{\arg\min}\left\{\rm{log}\left(\frac{1}{\textit{n}}
\sum\limits_{\textit{i}=1}^{\textit{n}}\left(\textit{y}_{\textit{i}}-\left\langle \textit{X}_{\textit{i}},\hat{\textit{B}}(\lambda)\right\rangle\right)^{2}\right)+\hat{\rm{df}}_{\lambda}\cdot\frac{\rm{log}(\textit{n})}{\textit{n}}\right\}$.
\end{center}

To prove the selection result of this new defined BIC, we introduce some technique conditions and new notations as follows.
\begin{itemize}
\item{(C1) $\{X_{i}\}_{i=1}^{n}$ are sampled i.i.d. from $X$ and $\{y_{i}\}_{i=1}^{n}$ are sampled i.i.d. from $y$. The predictor $X$ is standardized with $\|X\|_{F}=1$. $X$ and $y$ have finite fourth order moments, i.e.,
    \begin{center}
    $\underset{j,k}\max\left\{\rm{E}\left(\textit{X}^{4}_{\textit{j,k}}\right)\right\}<\infty$ and E$\left(y^{4}\right)<\infty$.
    \end{center}}
\item{(C2) $\exists B^{*}$ such that $B^{*}\neq0$, rank$(B^{*})=r^{*}<$min$\{p_{1},p_{2}\}$ and
\begin{center}
E$\left(y_{i}|X_{1},X_{2},\cdots,X_{n}\right)=\left\langle X_{i},B^{*}\right\rangle, \quad i=1,2,\cdots,n.$
\end{center}
\begin{center}
Var$\left(y_{i}|X_{1},X_{2},\cdots,X_{n}\right)=\sigma^{2}, \quad i=1,2,\cdots,n.$
\end{center}
}
\item{(C3) $\Sigma=$E$\left(\rm{vec}\left(\textit{X}\right)\rm{vec}\left(\textit{X}\right)^{\top}\right)$ is positive definite, which means there is a $\kappa>0$ such that the smallest eigenvalue of  $\Sigma$ is larger than $\kappa$.}
\item{(C4) $X$ and $\epsilon$ are independent with $\epsilon\sim$N$(0,\sigma^{2})$.}
%\item{(C5) $\|\Lambda\|_{2}<1$ with $$vec(\Lambda)=\left((\tilde{V}_{\perp}\otimes\tilde{U}_{\perp})^{\top}\Sigma^{-1}(\tilde{V}_{\perp}\otimes\tilde{U}_{\perp})\right)^{-1}
%    \left((\tilde{V}_{\perp}\otimes\tilde{U}_{\perp})^{\top}\Sigma^{-1}(\tilde{V}\otimes\tilde{U})vec(I_{p_{1}\times p_{2}})\right),$$
%    where the singular value decomposition of $B^{*}$ is $B^{*}=\tilde{U}diag(b^{*})\tilde{V}^{\top}$ with $\tilde{U}\in\mathbb{R}^{p_{1}\times r^{*}}$, $\tilde{V}\in\mathbb{R}^{p_{2}\times r^{*}}$, $\tilde{U}_{\perp}\in\mathbb{R}^{p_{1}\times(p_{1}-r^{*})}$ and $\tilde{V}_{\perp}\in\mathbb{R}^{p_{2}\times(p_{2}-r^{*})}$ are orthogonal complements  of $\tilde{U}$ and $\tilde{V}$, respectively.}
\end{itemize}
\begin{Remark}\label{remark3.1}
Condition (C2) states that there exists a low rank matrix $B^{*}$ such that the relationship of $X$ and $y$ is linear, which is a common assumption and ensures that the true solution has low rank. Conditions (C1) and (C3) are nature sampling assumptions. Condition (C4) guarantees that the random error is independent with the predictor.
%%Condition (C5) is a special assumption to ensure the consistency and rank consistency of the model (\ref{eq1}).
%Under conditions (C2) and (C3),  Bach \cite[Lemma 1]{B08} claimed that $\Sigma\in\mathbb{R}^{p_{1}\times p_{2}}$ satisfies
%\begin{center}
%E$\left(\left\|\frac{1}{n}\sum\limits_{i=1}^{n}\rm{vec}\left(\textit{X}_{\textit{i}}\right)\rm{vec}\left(\textit{X}_{\textit{i}}\right)^{\top}-\Sigma\right\|^{2}_{F}\right)
%=O\left(\frac{1}{n}\right)$ and $\frac{1}{\sqrt{n}}\sum\limits_{i=1}^{n}\epsilon_{i}\rm{vec}\left(\textit{X}_{\textit{i}}\right)\underset{\textit{d}}{\rightarrow}N(0,\sigma^{2}\Sigma)$.
%\end{center}

Under conditions (C1)-(C4), Bach \cite[Theorem 15]{bach2008consistency} proved the consistency and rank consistency of the solution of (\ref{eq1}). That is, if $\lambda_{n}n^{1/2+\gamma/2}\rightarrow\infty$ and $\lambda_{n}n^{1/2}\rightarrow0$ with $n\rightarrow\infty$,
$\underset{n\rightarrow\infty}{\rm{lim}}$P$\left(\hat{\textit{B}}(\lambda_{n})=B^{*}\right)=1$ and $\underset{n\rightarrow\infty}{\rm{lim}}$P$($r$_{\lambda_{n}}=r^{*})=1$.
\end{Remark}

Let $\Omega_{-}$ denote the underfitted case, which means $\Omega_{-}=\{\lambda:0<r_{\lambda}<r^{*}\}$. Let $\Omega_{+}$ denote the overfitted case, which means $\Omega_{+}=\{\lambda:r_{\lambda}>r^{*}\}$. Next, we prove the rank selection consistency of BIC$_{\lambda}$.
\begin{Theorem}\label{theorem3.1}
Assume technical conditions (C1)-(C4). Let $\lambda_{n}n^{1/2+\gamma/2}\rightarrow\infty$ and $\lambda_{n}n^{1/2}\rightarrow0$ with $n\rightarrow\infty$. For any $\lambda\in \Omega_{+}\bigcup\Omega_{-}$,
\begin{center}
$\underset{n\rightarrow\infty}{\rm{lim}}$P$\left(\rm{BIC}_{\lambda}>\rm{BIC}_{\lambda_{\textit{n}}}\right)=1$,
\end{center}
\end{Theorem}
\begin{proof}
For any $\lambda\in \Omega_{+}\bigcup\Omega_{-}$,
\begin{equation*}
\begin{aligned}
\rm{BIC}_{\lambda}-\rm{BIC}_{\lambda_{\textit{n}}}
=&\underset{a}{\underbrace{\rm{log}\left(\frac{1}{\textit{n}}
\sum\limits_{\textit{i}=1}^{\textit{n}}\left(\textit{y}_{\textit{i}}-\left\langle \textit{X}_{\textit{i}},\hat{\textit{B}}(\lambda)\right\rangle\right)^{2}\right)-\rm{log}\left(\frac{1}{\textit{n}}\sum\limits_{\textit{i}=1}^{\textit{n}}
\left(\textit{y}_{\textit{i}}-\left\langle \textit{X}_{\textit{i}},\hat{\textit{B}}(\lambda_{n})\right\rangle\right)^{2}\right)}}+\underset{b}{\underbrace{\hat{\rm{df}}_{\lambda}\cdot\frac{\rm{log}(\textit{n})}{\textit{n}}
-\hat{\rm{df}}_{\lambda_{n}}\cdot\frac{\rm{log}(\textit{n})}{\textit{n}}}}.
\end{aligned}
\end{equation*}

For a, we know that
\begin{equation*}
\begin{aligned}
a&=\rm{log}\left(\frac{1}{\textit{n}}
\left\|\textbf{y}-\mathcal{X}\rm{vec}\left(\hat{\textit{B}}(\lambda)\right)\right\|^{2}\right)-
\rm{log}\left(\frac{1}{\textit{n}}
\left\|\textbf{y}-\mathcal{X}\rm{vec}\left(\hat{\textit{B}}(\lambda_{\textit{n}})\right)\right\|^{2}\right)\\
&= \rm{log}\left(\frac{\left\|\textbf{y}-\mathcal{X}\rm{vec}(\hat{\textit{B}}(\lambda))\right\|^{2}/\textit{n}}
{\left\|\textbf{y}-\mathcal{X}\rm{vec}\left(\hat{\textit{B}}(\lambda_{\textit{n}})\right)\right\|^{2}/\textit{n}}\right)
=\rm{log}\left(\frac{\|\textbf{y}-\mathcal{X}\rm{vec}\left(\hat{\textit{B}}(\lambda_{\textit{n}})\right)+\mathcal{X}\rm{vec}\left(\hat{\textit{B}}(\lambda_{\textit{n}})-\hat{\textit{B}}(\lambda)\right)\|^{2}/\textit{n}}
{\left\|\textbf{y}-\mathcal{X}\rm{vec}\left(\hat{\textit{B}}(\lambda_{\textit{n}})\right)\right\|^{2}/\textit{n}}\right)\\
&= \rm{log}\left(1+\underset{\textit{a}1}{\underbrace{\frac{\rm{vec}\left(\hat{\textit{B}}(\lambda_{\textit{n}})-\hat{\textit{B}}(\lambda)\right)^{\top}\left(\mathcal{X}^{\top}\mathcal{X}/\textit{n}\right)\rm{vec}\left(\hat{\textit{B}}(\lambda_{\textit{n}})-\hat{\textit{B}}(\lambda)\right)}
{\left\|\textbf{y}-\mathcal{X}\rm{vec}(\hat{\textit{B}}(\lambda_{\textit{n}}))\right\|^{2}/\textit{n}}}}
+\underset{\textit{a}2}{\underbrace{\frac{2\left(\textbf{y}-\mathcal{X}\rm{vec}\left(\hat{\textit{B}}(\lambda_{\textit{n}})\right)\right)^{\top}\mathcal{X}\rm{vec}\left(\hat{\textit{B}}(\lambda_{\textit{n}})-\hat{\textit{B}}(\lambda)\right)/\textit{n}}
{\left\|\textbf{y}-\mathcal{X}\rm{vec}\left(\hat{\textit{B}}(\lambda_{\textit{n}})\right)\right\|^{2}/\textit{n}}}}\right)
\end{aligned}
\end{equation*}
%\begin{equation*}
%\begin{aligned}
%(a)&=\rm{log}\left(\frac{1}{\textit{n}}\sum\limits_{\textit{i}=1}^{\textit{n}}\left(\textit{y}_{\textit{i}}-\langle \textit{X}_{\textit{i}},\hat{\textit{B}}(\lambda)\rangle\right)^{2}\right)-\rm{log}\left(\frac{1}{\textit{n}}\sum\limits_{\textit{i}=1}^{\textit{n}}
%\left(\textit{y}_{\textit{i}}-\langle \textit{X}_{\textit{i}},\hat{\textit{B}}(\lambda_{\textit{n}})\rangle\right)^{2}\right)\\
%&=\rm{log}\left(\frac{\|\textbf{y}-\mathcal{X}\rm{vec}(\hat{\textit{B}}(\lambda))\|^{2}}
%{\|\textbf{y}-\mathcal{X}\rm{vec}(\hat{\textit{B}}(\lambda_{\textit{n}}))\|^{2}}\right)\\
%&=\rm{log}\left(1+\frac{\rm{vec}(\hat{\textit{B}}(\lambda)-\hat{\textit{B}}(\lambda_{\textit{n}}))^{\top}\mathcal{X}^{\top}\mathcal{X}\rm{vec}(\hat{\textit{B}}(\lambda)-\hat{\textit{B}}(\lambda_{\textit{n}}))/\textit{n}}{\|\textbf{y}-\mathcal{X}\rm{vec}(\hat{\textit{B}}(\lambda_{\textit{n}}))\|^{2}/\textit{n}}\right)
%\end{equation*}
For any $x$, we have log$(1+x)\geq\min\{\rm{log}(2),0.5\textit{x}\}$. So, $a\geq\min\{\rm{log}$(2),$0.5(a1+a2)\}$. Next, we compute these values. Remark \ref{remark3.1} illustrates that $P\left(\hat{\textit{B}}(\lambda_{\textit{n}})=B^{*}\right)\rightarrow1$ when $n\rightarrow\infty$, which leads to
\begin{center}
$\rm{vec}\left(\hat{\textit{B}}(\lambda_{\textit{n}})-\hat{\textit{B}}(\lambda)\right)^{\top}\frac{\mathcal{X}^{\top}\mathcal{X}}{\textit{n}}\rm{vec}\left(\hat{\textit{B}}(\lambda_{\textit{n}})-\hat{\textit{B}}(\lambda)\right)
\rightarrow$vec$\left(B^{*}-\hat{\textit{B}}(\lambda)\right)^{\top}
E\left(\rm{vec}(\textit{X})\rm{vec}(\textit{X})^{\top}\right)$ vec$\left(B^{*}-\hat{\textit{B}}(\lambda)\right)$.
\end{center}
Based on the fact that $\textbf{y}=\mathcal{X}\rm{vec}(\textit{B}^{*})+\varepsilon$ and the condition (C4),
\begin{equation*}
\begin{aligned}
\frac{\left\|\textbf{y}-\mathcal{X}\rm{vec}\left(\hat{\textit{B}}(\lambda_{\textit{n}})\right)\right\|^{2}}{n}
&=\frac{\left\|\mathcal{X}\rm{vec}\left(B^{*}-\hat{\textit{B}}(\lambda)\right)\right\|^{2}}{n}
+\frac{2}{n}\varepsilon^{\top}\mathcal{X}\rm{vec}\left(\textit{B}^{*}-\hat{\textit{B}}(\lambda)\right)+\frac{\varepsilon^{\top}\varepsilon}{\textit{n}}\rightarrow \sigma^{2}.
\end{aligned}
\end{equation*}
Combing these facts and (C3), we have $P\left(a1>\kappa\left\|B^{*}-\hat{\textit{B}}(\lambda)\right\|^{2}_{F}/\sigma^{2}\right)\rightarrow1$
and P$\left(a2=0\right)\rightarrow1$  when $n\rightarrow\infty$, which leads to
\begin{center}
$P\left(a\geq\min\left\{\rm{log}(2),\frac{\kappa}{2\sigma^{2}}\left\|\textit{B}^{*}-\hat{\textit{B}}(\lambda)\right\|^{2}_{\textit{F}}\right\}\right)\rightarrow1$ with $n\rightarrow\infty$.
\end{center}

For b, we consider the bound of $\hat{\rm{df}}_{\lambda}$. By Cauchy-Schwartz inequality, it is sure that
\begin{equation*}
\begin{aligned}
\hat{\rm{df}}_{\lambda}&=\frac{1}{\textit{n}}\sum\limits_{\textit{k}=1}^{\textit{n}}
\rm{vec}\left(\textit{X}_{\textit{k}}\right)^{\top}\textit{M}_{\textit{r}}^{+}\rm{vec}(\textit{X}_{\textit{k}})=\left\langle \frac{1}{\textit{n}}\sum\limits_{\textit{k}=1}^{\textit{n}}\rm{vec}\left(\textit{X}_{\textit{k}}\right)
\rm{vec}\left(\textit{X}_{\textit{k}}\right)^{\top},\textit{M}_{\textit{r}}^{+}\right\rangle\\
&\leq \left\|\frac{1}{\textit{n}}\sum\limits_{\textit{k}=1}^{\textit{n}}\rm{vec}\left(\textit{X}_{\textit{k}}\right)
\rm{vec}\left(\textit{X}_{\textit{k}}\right)^{\top}\right\|_{F}\cdot\left\|\textit{M}_{\textit{r}}^{+}\right\|_{2}\leq \frac{1}{\textit{n}}\sum\limits_{\textit{k}=1}^{\textit{n}}\left\|\textit{X}_{\textit{k}}\right\|^{2}_{F}\cdot \left\|\textit{M}_{\textit{r}}^{+}\right\|_{2}.
\end{aligned}
\end{equation*}
When $n\rightarrow\infty$, $\frac{1}{n}\sum\limits_{k=1}^{n}\left\|X_{k}\right\|^{2}_{F}\rightarrow E(\|X\|^{2}_{F})$, which is bounded because of (C1). In addition, $\left\|\textit{M}_{\textit{r}}^{+}\right\|_{2}
=\frac{1}{\sigma_{\tilde{r}}\left(\textit{M}_{\textit{r}}\right)}$ with $\tilde{r}$ being the rank of $\textit{M}_{\textit{r}}$, which is bounded when $n\rightarrow\infty$. Combining the fact that $\frac{\rm{log}(\textit{n})}{\textit{n}}\rightarrow0$,  we have  log$(2)+b>0$ and $\frac{\kappa}{2\sigma^{2}}\left\|\textit{B}^{*}-\hat{\textit{B}}(\lambda)\right\|^{2}_{F}+b>0$ with probability tending to one. Therefore, the result about $P\left(\rm{BIC}_{\lambda}-\rm{BIC}_{\lambda_{\textit{n}}}>0\right)\rightarrow1$ holds,  when $n\rightarrow\infty$.
\end{proof}

According to this theorem, the information criterion BIC$_{\lambda}$ will not select the underfitted or overfitted models when the sample size increasing, which means BIC$_{\lambda}$ can select the true model consistently.
\section{Numerical Experiments}
This section will show some numerical results of our proposed Bayesian information criterion (BIC) for tuning parameter selection. To compare our method with the popular tuning parameter selection methods, we extend the corresponding expressions for the model (\ref{eq1}) without further proofs, which includes Akaike Information Criterion (AIC),  Akaike Information Criterion corrected (AICc) and cross validation. The difference between BIC, AIC and AICc is on the second term of their expressions, which are showed as follows.

AIC is defined as
\begin{eqnarray*}
\rm{AIC}_{\lambda}=\rm{log}\left(\sum\limits_{\textit{i}=1}^{\textit{n}}
\frac{\left(\textit{y}_{\textit{i}}-\left\langle \textit{X}_{\textit{i}},\hat{\textit{B}}(\lambda)\right\rangle\right)^{2}}{\textit{n}}\right)+\hat{\rm{df}}_{\lambda}\cdot\frac{2}{\textit{n}}.
\end{eqnarray*}

AICc is defined as
\begin{eqnarray*}
	\rm{AICc}_{\lambda}=\rm{log}\left(\sum\limits_{\textit{i}=1}^{\textit{n}}
\frac{\left(\textit{y}_{\textit{i}}-\left\langle \textit{X}_{\textit{i}},\hat{\textit{B}}(\lambda)\right\rangle\right)^{2}}{\textit{n}}\right)+\hat{\rm{df}}_{\lambda}
\cdot\frac{2}{\textit{n}}+\frac{2\hat{\rm{df}}_{\lambda}\cdot(\hat{\rm{df}}_{\lambda}+1)}{\textit{n}-\hat{\rm{df}}_{\lambda}-1}.
\end{eqnarray*}

$K$-fold cross-validation is a popular method used to evaluate and compare the generalization ability to predict new data. This approach works by splitting the dataset into $K$ smaller subsets (folds). One of these folds is used as a test set to validate the performance of the model, while the other $K-1$ subsets are used as training sets for the model. This process is repeated $K$ times, and each time selects a different fold as the test set and the rest as the training set. Finally, the $K$ test results are averaged or otherwise integrated to obtain an overall model performance evaluation. The usual choice of $K$ is 5 and 10, which are 5-fold cross validation and 10-fold cross validation, respectively.

In the following numerical experiments, we compare our BIC with AIC, AICc, 5-fold cross validation and 10-fold cross validation. The comparison indexes include the selected best tuning parameter $\lambda^{*}$, the mean square error (MSE) on the predictor $y$, the rank of the selected solution $r$ and the computational time. Here, the MSE is defined as $\|\hat{y}-y\|^{2}$, where $\hat{y}$ is the predicted response variable under the selected tuning parameter. Our numerical experiments are all implemented on MATLAB R2023b with AMD Ryzen 5 5600H with Radeon Graphics  3.30 GHz and 16G RAM.
\begin{table}[htbp]
\caption{Simulation results on prediction matrix dimension as $p_{1}=15$ and $p_{2}=45$. The sample size includes $10^{3}$, $10^{2}$ and $50$, respectively.}
\begin{tabular}{ll|llll}
\hline
sample size                 & method     & $\lambda^{*}$ & MSE & r & time (s) \\ \hline
\multirow{5}{*}{$n=10^{3}$} & BIC        & 22.860        & 7.2380    & \textbf{2}  &  \textbf{101.69}        \\
                            & AIC        & 0.0600        & 0.0808    & 6  &   \textbf{101.69}       \\
                            & AIC$_{c}$  & 0.0003        & 150.23    & 15 &  \textbf{101.69}        \\
                            & 5-fold CV  & 0.1361        & 0.0259    & \textbf{2}  &   178.12        \\
                            & 10-fold CV & 0.1361        & 0.0259    & \textbf{2}  &   672.54       \\ \hline
\multirow{5}{*}{$n=10^{2}$} & BIC        & 0.8335        & 0.0154 & \textbf{4} & \textbf{121.22} \\
		                    & AIC        & 0.8335       & 0.0154  & \textbf{4} & \textbf{121.22} \\
		                    & AICc       & 0.0004       & 0.0002   & 15         & \textbf{121.22} \\
		                    & 5-fold CV  & 0.0006        & 1.22$\times 10^{-8}$ & 14    & 344.23           \\
		                    & 10-fold CV & 0.0006       & 1.22$\times 10^{-8}$ & 14      & 5112.1         \\ \hline
\multirow{5}{*}{$n=50$}     & BIC        & 6.8011       & 0.0154  & \textbf{4} & \textbf{110.38} \\
		                    & AIC        & 6.8011       & 0.0154  & \textbf{4} & \textbf{110.38} \\
		                    & AICc       & 0.0906       & 0.0002 & \textbf{4} & \textbf{110.38} \\
		                    & 5-fold CV  & 0.0026       & 1.22$\times 10^{-8}$ & 10         & 344.38          \\
		                    & 10-fold CV & 0.0026       & 1.22$\times 10^{-8}$ & 10         & 484.55         \\ \hline
\end{tabular}
\end{table}
\subsection{Simulation Data}
As in Bach \cite{bach2008consistency}, we simulate a lots of data sets with different values $p_{1}$, $p_{2}$ and $n$, where the prediction matrixes follow the standard normal distribution/Gaussian distribution. Here, we select $p_{1}\in\{15,25,35\}$, $p_{2}\in\{30,45\}$ and $n\in\{10^{3},10^{2},50\}$. We generate random i.i.d. data matrixes $P_{i}\in\mathbb{R}^{p_{1}}$ and $Q_{i}\in\mathbb{R}^{p_{2}}$, and we set a true solution $B^{*}\in\mathbb{R}^{p\times q}$ with rank 2. All elements of these matrix  distribute with the standard normal distribution. Then,  $y_{i}=\langle P_{i}Q^{\top}_{i},B^{*}\rangle + \epsilon_{i},i=1,2,\cdots,n$, where $\epsilon_{i}$ have i.i.d components with normal distributions with zero mean and 0.1 standard variance.  To select the best tuning parameter, same as the setting in Bach \cite{bach2008consistency}, we set the tuning parameter sequence as $\lambda_{k}=$e$^{\rm{log}(\lambda_{\max})+(k-1)\times\frac{\rm{log}(\lambda_{\min})-\rm{log}(\lambda_{\max})}{100}}$ with $k=1,2,\cdots,100$.  The numerical results are reported in Tables 1 to 3.
\begin{table}[htbp]
\caption{Simulation results on prediction matrix dimension as $p_{1}=25$ and $p_{2}=30$. The sample size includes $10^{3}$, $10^{2}$ and $50$, respectively.}
\begin{tabular}{ll|llll}
\hline
sample size                 & method     & $\lambda^{*}$ & MSE & r & time (s) \\ \hline
\multirow{5}{*}{$n=10^{3}$} & BIC         & 0.1092 & 0.1210  & \textbf{3} & \textbf{56.763} \\
		                    & AIC         & 0.0218 & 0.0487 & 7& \textbf{56.763} \\
		                    & AICc        & 0.0489 & 0.0569 & 4 & \textbf{56.763} \\
		                    & 5-fold CV  & 0.0218 & 0.0487& 7        & 128.60        \\
		                    & 10-fold CV & 0.0218 & 0.0487 & 7        & 266.08      \\ \hline
\multirow{5}{*}{$n=10^{2}$} & BIC         & 0.0176 & 0.0001 & 14 & \textbf{89.548} \\
		                    & AIC         & 0.0176 & 0.0001 & 14& \textbf{89.548} \\
		                    & AICc        & 0.0596 & 0.0003 & \textbf{10} & \textbf{89.548} \\
		                    & 5-fold CV  & 0.0176 & 0.0001& 14       & 301.95        \\
		                    & 10-fold CV & 0.0176 & 0.0001 & 14        & 802.78      \\\hline
\multirow{5}{*}{$n=50$}     & BIC         & 0.6608 & 0.0629 & \textbf{6} & \textbf{112.48} \\
		                    & AIC         & 0.6608 & 0.0629 & \textbf{6}& \textbf{112.48} \\
		                    & AICc        & 0.0277 & 0.3611 & 12 & \textbf{112.48} \\
		                    & 5-fold CV  & 0.0910 & 0.0292& 7       & 259.31        \\
		                    & 10-fold CV & 0.0612 & 0.0279 & 9        & 498.99      \\\hline
\end{tabular}
\end{table}
\begin{table}[htbp]
\caption{Simulation results on prediction matrix dimension as $p_{1}=35$ and $p_{2}=30$. The sample size includes $10^{3}$, $10^{2}$ and $50$, respectively.}
\begin{tabular}{ll|llll}
\hline
sample size                 & method     & $\lambda^{*}$ & MSE & r & time (s) \\ \hline
\multirow{5}{*}{$n=10^{3}$} & BIC         & 0.0089 & 0.0134  & \textbf{13} & \textbf{190.15} \\
		                    & AIC         & 0.0026 & 0.0050  & 22          & \textbf{190.15} \\
		                    & AICc        & 0.0059 & 0.0088 & 16           & \textbf{190.15} \\
		                   & 5-fold CV    & 0.0007 & 0.0025 & 30         & 566.70        \\
		                   & 10-fold CV   & 0.0007 & 0.0025 & 30        & 1113.4      \\ \hline
\multirow{5}{*}{$n=10^{2}$}& BIC         & 0.4744 & 0.2561  & \textbf{10} & \textbf{178.11} \\
		                   & AIC         & 0.4744 & 0.2561 & \textbf{10}& \textbf{178.11} \\
		                   & AICc        & 0.3883 & 0.3430 & \textbf{10} & \textbf{178.11} \\
		                   & 5-fold CV  & 0.2602 & 0.1944  & 11         & 655.69        \\
		                   & 10-fold CV & 0.2602 & 0.1944  & 11        & 1125.49      \\ \hline
\multirow{5}{*}{$n=50$}    & BIC         & 1.2896 & 66.480  & \textbf{7} & \textbf{158.60} \\
		                   & AIC         & 1.2896 & 66.480 & \textbf{7}& \textbf{158.60} \\
		                   & AICc        & 0.0009 & 158.11 & 25 & \textbf{158.60} \\
		                   & 5-fold CV  & 1.0588 & 58.339& \textbf{7} & 448.32        \\
		                   & 10-fold CV & 1.0588 & 58.338 & \textbf{7} & 1397.4      \\ \hline
\end{tabular}
\end{table}

Based on these results in Tables 1 to 3, there are some conclusions. (i) Comparing to the other tuning parameter methods, our proposed BIC will select the best tuning parameter with the smallest solution rank. Because the true rank of the solution is 2,  BIC will tend to select the best tuning parameter with the same true solution rank. (ii) The computational time of BIC is same with AIC and AICc, and obviously smaller than that of  5-fold CV and 10-fold CV. Because BIC, AIC and AICc are different only on the second term, their computational time are same. Comparing to cross validation, our BIC surely reduce the computational time, because that our method  does not need to solve the model many times as cross validation. (iii)  BIC generally performs same with the other  methods on MSE and  the MSE under BIC is larger than cross validation in some special cases, which is a trade off of the computational effect and computational cost.

\subsection{Real Data}
This section applies the tuning parameter selection process of the model (\ref{eq1}) on the COVID-19 data set and Bike data set, and compares BIC with AIC, AICc, 5-fold CV and 10-fold CV in Table 4 and Table 5. For every data set, we standardize the data set by columns, i.e., every number is subtracted by the column mean and divided by the column standard deviation. To select the best tuning parameter, we set the tuning parameter sequence as $\lambda_{k}=0.618^{k}\lambda_{\max}$ with $k=1,2,\cdots,20$ and $\lambda_{\max}=\left\|W^{\top}_{1}\cdots\left(\sum\limits_{i=1}^{n}y_{i}X_{i}\right)
W^{\top}_{2}\right\|_{2}$. Here, $\lambda_{\max}$ is the smallest tuning parameter such that the solution of (\ref{eq1}) is zero.

The COVID-19 dataset (\cite{li2021double, wahltinez2020covid})  consists of daily measurements related to COVID-19 for 138 countries around the world. This data set records the newly confirmed case in the period June 13, 2020 to July 12, 2020. In addition, this data also includes the 41 COVID-19 related government policies in each day, i.e.,  school-closing, restrictions on gathering, stay-at-home requirement, income support and so on. Each of these policies may have several levels, for example, school closing includes no closing, recommend closing, require some closing (e.g. just high school) or require all closing, which varied during the 30-day period. Therefore, for every sample, its prediction matrix is  $X\in\mathbb{R}^{41 \times 30}$  and response is the newly confirmed case. The sample size is $n = 138$.
\begin{table}[htbp]
	\caption{COVID data set results with $p_{1}=41$, $p_{2}=30$ and $n=138$.}
	\centering
	\begin{tabular}{lllll}
		\hline
		Method      & $\lambda^{*}$   & MSE      & rank       & time            \\
		\hline
		BIC         &0.1346 & 0.0231 & \textbf{1} & \textbf{19.088} \\
		AIC         &0.1346 & 0.0231 & \textbf{1}& \textbf{19.088} \\
		AICc        &0.1346 & 0.0231 & \textbf{1} & \textbf{19.088} \\
		5-fold CV   &0.1346 & 0.0231 & \textbf{1} & 37.855     \\
		10-fold CV  &0.1346 & 0.0231 & \textbf{1} & 74.117     \\
		\hline
	\end{tabular}
\end{table}

The Bike sharing data set \cite{fanaee2014event} includes matrixes data from January 1st, 2011 to December 31th, 2012. For every data during this period, the recorded data includes 6 weather conditions (including sunny, mist or others, temperature, apparent temperature, humidity and wind speed) every hour, which makes the prediction matrix dimension as $X\in\mathbb{R}^{24 \times 6}$ and the sample size $n=731$. The response variable is the daily aggregated count of rented bikes.
\begin{table}[htbp]
\caption{Bike sharing data set results with $p_{1}=24$, $p_{2}=6$ and $n=731$.}
\centering
\begin{tabular}{lllll}
	\hline
	Method      & $\lambda^{*}$   & MSE      & rank       & time           \\
	\hline
	BIC         & 1.5074 & 0.5173 & \textbf{1} & \textbf{1.6880} \\
	AIC         & 0.3557 & 0.4911 & 2& \textbf{1.6880} \\
	AICc        & 0.0011 & 0.4029 & 6 & \textbf{1.6880} \\
	5-fold CV   & 0.0011 & 0.4029 & 6 &7.1090     \\
   10-fold CV  & 0.0011 & 0.4029 & 6  & 9.2680      \\
	\hline
\end{tabular}
\end{table}

On these real data sets, our proposed BIC performs better than AIC and AICc because that the rank under $\lambda^{*}$ with BIC is the smallest and their MSE values are similar. In addition, BIC outperforms 5-fold CV and 10-fold CV on computational time.
\section{Conclusion}

To select the best tuning parameter of the adaptive nuclear norm regularized trace regression model, we propose a Bayesian information criterion (BIC) in this paper. We first compute the unbiased  estimator of the degrees of the freedom of the model, which is the basic element of the BIC. Under this estimator, we build up the BIC and prove its rank selection consistency, i.e., BIC will select the tuning parameter achieving the true solution and true solution rank in probability with the sample size increasing. To evaluate the  performance of BIC, we compare it with  AIC, AICc, 5-fold CV and 10-fold CV on some simulation data and real data sets. The numerical results show that BIC performs better than the other methods.
\section*{Acknowledgements}
This work was supported by the National Natural Science Foundation of China (12371322).



\bibliographystyle{plain}
\bibliography{Bib,refs}
\end{document}